\documentclass[12pt]{spieman}  
\usepackage{hyperref}
\usepackage{amsmath,amsfonts,amssymb}
\usepackage{setspace}
\usepackage{tocloft}
\usepackage{amscd}
\usepackage{upgreek}
\usepackage[mathscr]{eucal}
\usepackage{graphicx}
\usepackage{algorithm}
\usepackage{algorithmic}
\usepackage{caption}
\usepackage{subcaption}
\usepackage{color}
\usepackage{rotating}
\usepackage{placeins}
\usepackage{float}
\usepackage{booktabs}
\usepackage{tabularx}
\usepackage{bbding}
\usepackage[misc]{ifsym}
\usepackage{array}
\usepackage[modulo]{lineno}
\usepackage{lipsum}
\usepackage{color,soul}
\soulregister\eqref{7}
\soulregister\cite{7}
\soulregister\mathbf{7}

\linenumbers
\newtheorem{prop}{Proposition}

\title{Super-resolving multiple scatterers detection in SAR Tomography assisted by correlation information}

\author[a,*]{Ahmad Naghavi}
\author[a]{Mohammad Sadegh Fazel}
\author[b]{Mojtaba Beheshti}
\author[a]{Ehsan Yazdian}
\affil[a]{Isfahan University of Technology, Department of Electrical and Computer Engineering, Isfahan, Iran}
\affil[b]{Research Institute for Information and Communication Technology, Isfahan University of Technology, Isfahan, Iran}

\cftpagenumbersoff{figure}
\cftpagenumbersoff{table} 
\begin{document} 
\maketitle

\begin{abstract}
This paper proposes a method for detecting multiple scatterers (targets) in the elevation direction for synthetic aperture radar (SAR) tomography. The proposed method can resolve closely spaced targets through a two-step procedure. In the first step, coarse detection is performed with a successive cancellation scheme in which possible locations of targets are marked. Then, in the second step, by searching in the reduced search space which is finely gridded, the accurate location of the targets is found. For estimating the actual number of targets, a model order selection scheme is used in two cases of known and unknown noise variance. Also, by analytical investigation of the probability of detection for the proposed method, the effect of the influential parameters on the detection ability is explicitly demonstrated. Compared to the super-resolution methods based on compressed sensing (CS), the proposed method has a lower computational cost and higher estimation accuracy, especially at low signal-to-noise ratio regime. Simulation results show the superiority of the proposed method in terms of both 3D scatterer reconstruction and detection ability.
\end{abstract}


{\noindent \footnotesize\textbf{*}Ahmad Naghavi,  \linkable{a.naghavi@ec.iut.ac.ir} }

\begin{spacing}{2}   

\section{Introduction}\label{sec1}
{S}{ynthetic} aperture radar (SAR) facilitates a day-night all-weather two-dimensional imaging for earth observation \cite{moreira2013tutorial}. Space-borne missions such as TerraSAR-X (TSX) \cite{eineder2009spaceborne} and COSMO-SkyMed \cite{budillon2014localization} capture very high resolution (VHR) data stacks. Using three-dimensional SAR imaging, the distribution of scatterers in the elevation direction is also achieved. Synthetic aperture radar interferometry (InSAR) is a powerful tool for height profile estimation but it can not resolve the scatterers in height \cite{bamler1998synthetic}. 

As a complementary method, SAR tomography (TomoSAR) is capable of resolving the scatterers through a synthetic aperture in the elevation direction\cite{reigber2000first}. It retrieves the location of the scatterers and their reflectivity by processing a stack of range-azimuth focused images provided by captured data from slightly different orbits \cite{fornaro2005three}. Monitoring applications of SAR tomography are in two main categories: natural land coverage and urban infrastructure. Furthermore, by using multi-temporal SAR acquisition, TomoSAR has been extended to higher dimensions. Four-dimensional (4D) tomography presents deformation velocity of the scatterers in addition to elevation \cite{lombardini2005differential,fornaro2008four}. By five-dimensional (5D) tomography, thermal expansion of the scatterers is also obtained \cite{monserrat2011thermal,zhu2011let}.

The main processing in TomoSAR is the inversion, which is the extraction of the reflectivity vector of scatterers in elevation direction from the acquired data. This leads to a spectral estimation problem which can be solved by a parametric or non-parametric method. The first non-parametric solution is beamforming \cite{reigber2000first}. Nonlinear least square (NLS), as a parametric method that is accurate but is time-consuming \cite{zhu2010very}. Furthermore, the parametric methods such as NLS and need knowledge about the number of targets \cite{stoica1997introduction} and hence, should employ a model order selection scheme.

In addition to TomoSAR inversion methods using model order selection, detection-based methods have been proposed for scatterers detection \cite{de2009detection,pauciullo2012detection}. Since there are some unknown parameters in the detection process, a common solution is using the generalized likelihood ratio test (GLRT). In \cite{de2009detection}, GLRT is used for single scatterer detection. For the double scatterers case, the solution called sequential GLRT with cancellation (SGLRTC) is proposed which is fast but has low resolution\cite{pauciullo2012detection}. In \cite{budillon2016glrt}, another GLRT-based method called Sup-GLRT has been presented. This method is a high-resolution method and can detect more than two targets at the expense of searching in larger combinatorial spaces. 

Compressed sensing (CS) as a sparse data recovery method, has been widely used in signal processing \cite{donoho2006compressed}. To achieve super-resolution by CS, a limited number of observations suffices, and knowledge of the number of targets is not necessary. After compressed sampling, a reconstruction algorithm is needed to recover the data. Among the several algorithms, $l_1$-norm minimization has received the most attention \cite{chen2001atomic}\cite{tibshirani1996regression}. Despite the high accuracy of $l_1$-norm minimization algorithms, they are computationally complex and suffer from outliers after reconstruction. \cite{eldar2012compressed}\cite{donoho2005stable}. In various recent studies on TomoSAR, $l_1$-norm minimization has been applied to obtain height profile of urban infrastructures \cite{zhu2010tomographic,budillon2011three,zhu2012super}. In \cite{zhu2012super}, the Scale-down by L1 norm Minimization, Model selection, and Estimation Reconstruction (SL1MMER) method was proposed in which, by applying model order selection over the downsized data resulted by CS reconstruction, the number of targets and their locations are obtained.

While super-resolution capability is desirable in every imaging system, in TomoSAR, it is more critical because of the layover phenomenon. When layover occurs, multiple targets in elevation contribute to the same range-azimuth cell. By utilizing tomography, layover can be resolved if the minimum distance between the scatterers is more than one resolution unit in the elevation direction.

In this paper, we propose a computationally-efficient super-resolution method that accurately detects the scatterers through two consecutive steps. First, in the coarse detection step, by exploiting correlation information, the possible locations of targets are characterized. Then in the fine detection step, a combinatorial search is applied for locating the targets. Also, in this step, we used a model order selection scheme for the two cases of known and unknown noise variance. Detection ability and estimation accuracy of the proposed method are used as performance criteria and investigated regarding the related factors. Moreover, by deriving an analytical expression for the probability of detection of the proposed method, we investigate the impact of the influential factors on the detection ability.

The rest of this paper is organized as follows. The system model appears in Section \ref{sec2}. In Section \ref{sec3}, the proposed two-step detection method is presented. The performance criteria of the proposed method are analyzed in Section \ref{sec4}.  Section \ref{sec5} presents the results of TomoSAR experiments over simulated SAR data. Section \ref{sec6} states the advantages of the proposed method in computational cost. Finally, section \ref{sec7} concludes this work.
\section{System Model}\label{sec2}
The multipass SAR geometry is shown in Fig. \ref{fig:tomosar-geometry}. In TomoSAR processing, first a stack of complex datasets are acquired at $N$ orbit positions (orthogonal baselines). Then, range-azimuth compression is applied on the raw data. It is shown that (after some postprocessing), the focused complex image for a range-azimuth pixel at $n$-th pass is \cite{fornaro2005three}
\begin{equation}
\label{1}
g_n=\underset{\Delta_s}\int\pmb{\gamma}(s)\exp(j2{\pi}\xi_ns)ds,\hspace{0.15cm}n=0,\ldots,N-1
\end{equation}
where $\pmb{\gamma}(s)$ is the reflectivity function along the elevation direction $s$, $\xi_n={2b_n}/(\lambda R_0)$ is the spatial frequency, $b_n$ stands for the $n$-th orthogonal baseline relative to the reference (master) acquisition, $\lambda$ is operating wavelength, $R_0$ is the distance from the master baseline to the center of the ground scene and $\Delta_s$ is the elevation extent. Equation \eqref{1} shows a spatial Fourier transform relationship between the target reflectivity in the elevation direction and the  post-processed complex image at $n$-th pass. Discretizing \eqref{1} along $s$ dimension yields the following linear model
\begin{figure}
 \centering
 \includegraphics[width=.65\linewidth]{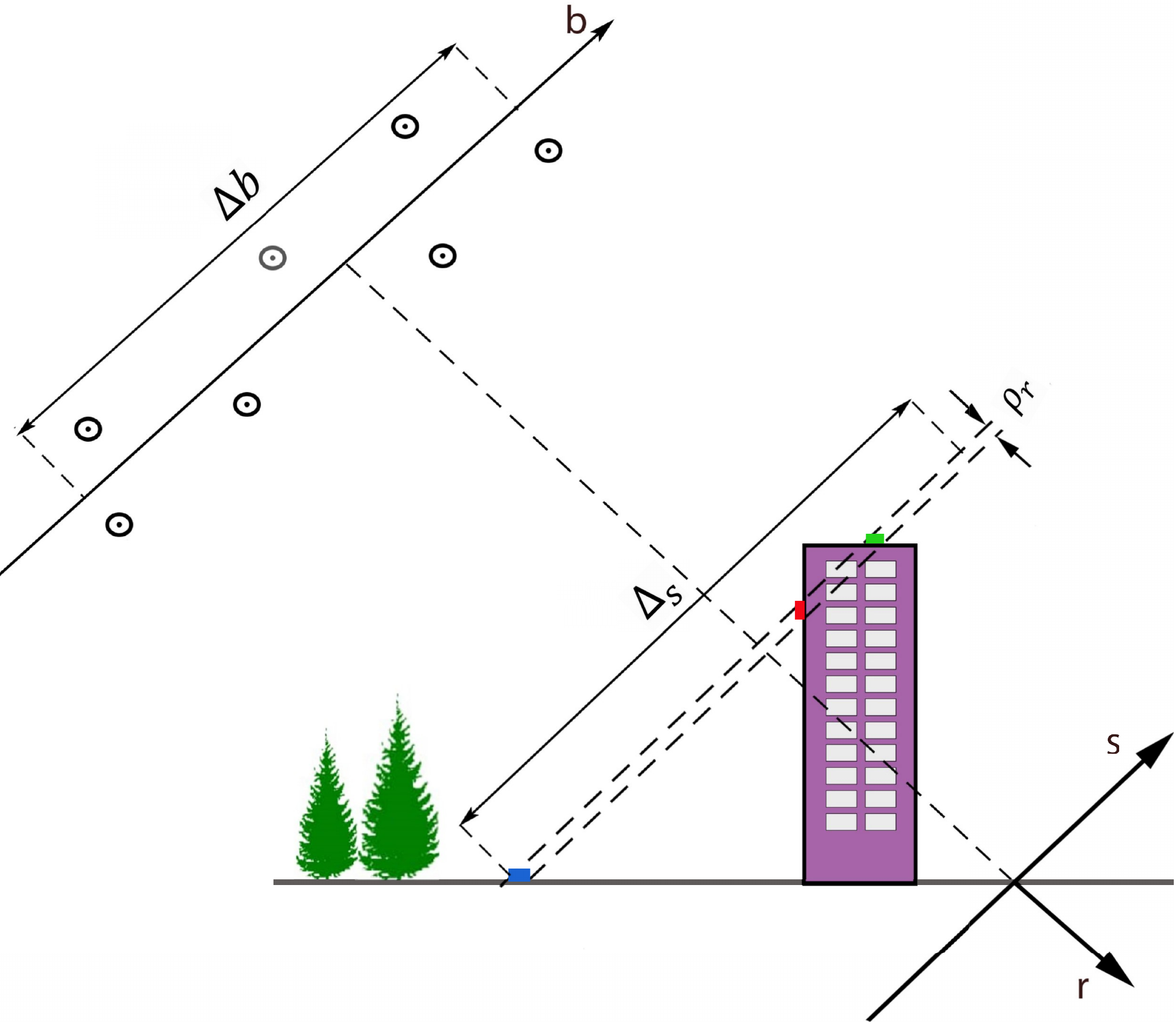}
 \caption{Multipass SAR geometry in the range-elevation $(r,s)$ plane.}
 \label{fig:tomosar-geometry}
\end{figure}
\begin{equation}
\mathbf{g}=\mathbf{A}\boldmath{{\gamma}}+\mathbf{n}
\label{2}
\end{equation} 
where $\mathbf{g}=\left[g_1, g_2,...,g_N\right]^\mathrm{T}$ is the vector of the range-azimuth compressed images obtained at all passes and $\mathrm{(.)^T}$ denotes the transpose operator. Let $M$ be the number of samples (grids) in the elevation direction. $s_m,m=1,2,...M$ is the $m'$th elevation samples and $\pmb{\gamma}=\left[\gamma_1, \gamma_2,\ldots,\gamma_M\right]^\mathrm{T}$ is the reflectivity vector.
The matrix  $\mathbf{A}=[\boldsymbol{a}(s_1),\boldsymbol{a}(s_2),... ,\boldsymbol{a}(s_M)]$ is the $N\times{M}$ steering matrix where $\boldsymbol{a}(s_m)$ is the m'th steering vector whose $n$'th element is $\exp(j2\pi\xi_{n}s_{m})$ and $j=\sqrt{-1}$. The noise term $\mathbf{n}$ is a zero-mean circular complex Gaussian vector with covariance matrix $E{\{\mathbf{n}\mathbf{n}^\mathrm{H}}\}=\sigma^2_n\mathbf{I}_N$, where $\mathbf{I}_N$ is $N\times{N}$ identity matrix.

The basic tomographic inversion method is beamforming \cite{reigber2000first} which is actually the correlation between $\mathbf{g}$ and the steering matrix $\mathbf{A}$. Accordingly, Rayleigh resolution is defined as the minimum separation between two point scatterers which can be discriminated by the beamforming method
\begin{equation}
\rho_{s}=\frac{{\lambda}R_0}{2\Delta b}
\label{3}
\end{equation} 
where $\Delta b$ is the total baseline extent. 
\section[]{Two-Step Multiple Scatterers Detection}\label{sec3}
In SAR tomography, by using an inversion method, all elements of the reflectivity vector are estimated. However, for point scatterers, the objective of TomoSAR is only the detection of non-zero elements of the reflectivity vector. GLRT and SGLRTC are two detectors  \cite{de2009detection,pauciullo2012detection} in which, detection is based on beamforming method or equivalently correlation information. Since the number of passes is usually small, the baseline extent $\Delta b$ is limited, and as a result, the resolution of the SGLRTC detector is low. A super-resolution method suitable for multiple scatterers detection is the NLS method, which is based on maximum likelihood (ML) estimation \cite{kay1993fundamentals}. The NLS method, though an accurate solution, has high computational complexity due to the required exhaustive search \cite{zhu2010very}.

To reduce the computational cost of NLS, a two-step method is proposed called Correlation-Aided NLS (CA-NLS). In its first step, coarse detection is performed through finding possible location of targets. Then by searching in the reduced space obtained in the first step, fine detection is achieved. In the following, the two-step method is presented in detail.
\subsection{Coarse Detection}\label{sec3_1}
In the coarse detection step, we use the correlation information in an iterative routine to characterize the possible location of the targets. To do this, first the correlation between the received data $\mathbf{g}$ and the steering matrix $\mathbf{A}$ is calculated. Then, the target is detected by comparing the maximum power of the correlation vector (after scaling) with a threshold as \cite{conte1995asymptotically,de2009detection}
\begin{equation}
\max_{s_m}\frac{|\boldsymbol{a}^H(s_m)\mathbf{g}|^2}{N\|\mathbf{g_{\perp}}\|^2}\underset{\mathcal{H}_0}{\overset{\mathcal{H}_1}{\gtrless}}T, \hspace{.3cm}m=1,2,...M 
\label{4}
\end{equation}
where $T$ is a threshold adjusted based on a constant false alarm rate (CFAR) scheme and $\mathbf{g_{\perp}}$ is the residual component of $\mathbf{g}$ that is orthogonal to $\boldsymbol{a}(s_m)$
\begin{equation}
\mathbf{g_{\perp}}= \mathbf{g} - \boldsymbol{a}(s_m)[\boldsymbol{a}(s_m)^H\boldsymbol{a}(s_m)]^{-1}\boldsymbol{a}(s_m)^H\mathbf{g}.
\label{5}
\end{equation}

For the multi-target case, the location of the detected target obtained by \eqref{4} does not necessarily coincide with their real location. However, the vicinity area of the detected peak can be characterized as the possible location of a target. This area has the extent of $2\rho_s$ and is centered around the detected peak.

The procedure of coarse detection is shown in Algorithm \ref{alg.CoarseLocalization} which includes two consecutive loops of iteration and decision. In the iteration loop, the index of the dominant components is found using \eqref{4}, and then, by canceling the detected dominant component, the residual term is updated. This routine which resembles orthogonal matching pursuit (OMP) \cite{tropp2007signal}, is successively done until the maximum number of targets $k_\mathrm{max}$ are detected. The updating parameters at $k$-th iteration are as follows: $\mathbf{r}_k$ is the residual term, $p_k$ is the index of the peak point of correlation vector between $\mathbf{r}_{k-1}$ and $\mathbf{A}$, $\Lambda_k$ is the vector of detected peaks indexes, $\Gamma_k$ is the test statistics of \eqref{4} and $supp_{k}$ is the partial support related to the detected peak. In the decision loop, the derived test statistics are compared with the threshold $T$ in reverse order. If in the $k$'th loop $\Gamma_k$ exceeds the threshold, $S$ which is the possible support of the targets will be the union of all $k$ detected partial supports. 
\begin{algorithm}[h!] 
 \caption{Coarse Detection} 
 \label{alg.CoarseLocalization} 
 \begin{algorithmic}[1] 
  
  \STATE $\mathbf{Input}: \mathbf{g},\mathbf{A},\rho_s,k_\mathrm{max},T $
  \STATE Initialization: $k=0,\mathbf{r}_0=\mathbf{g},\Lambda_0=\emptyset,supp_0=\emptyset$
  \STATE -- Iteration loop  \WHILE{$ \hspace{0.15cm} (k \textless k_\mathrm{max} ) $}
  \STATE$k=k+1$  
  \STATE$p_k=arg\max_{i}{\boldsymbol{|}\boldsymbol{a}^H(s_{i})\mathbf{r}_{k-1}\boldsymbol{|}}, \hspace{0.15cm} i=1,2,...,M$
  \STATE$\Lambda_k=\Lambda_{k-1}\bigcup{p_k}$
  \STATE$\mathbf{r}_{k}=\mathbf{g}-\mathbf{A}_{\Lambda_k}(\mathbf{A}^H_{\Lambda_k}\mathbf{A}_{\Lambda_k})^{-1}\mathbf{A}^H_{\Lambda_k}\mathbf{g}$   \STATE$\Gamma_k=\frac{\boldsymbol{|}\boldsymbol{a}^H(s_{p_k})\mathbf{r}_{k-1}\boldsymbol{|}^2}{N.\|\mathbf{r}_{k}\|^2}$
  \STATE$supp_k =\boldsymbol{[}s_{p_k}-\rho_s,s_{p_k}+\rho_s\boldsymbol{]}$
  \ENDWHILE 
  \STATE -- Decision loop
  \WHILE{$k > 0$}
  \IF{$\Gamma_k > T$}
  \STATE break
  \ELSE
  \STATE$k=k-1$ 
  \ENDIF
  \ENDWHILE
  \STATE $\mathbf{Output}$ : {$S = \bigcup\limits^{k}_{i=0} {supp_i} $}
 \end{algorithmic}
\end{algorithm}

For the special case of $k_\mathrm{max}=2$, the coarse detection algorithm presented is similar to the SGLRTC method \cite{pauciullo2012detection}. In our method, however, the support information acquired during the coarse detection step, enables fast detection of targets in the next step.
\subsection{Fine detection}\label{sec3_2}
After identification of the possible location of targets, it is desired to find the accurate location of targets. Assumming presence of $k$ targets, $\Omega_k$ is defined as support of $k$ targets (position of nonzero elements of reflectivity vector), $\pmb{\gamma}_{\Omega_{k}}$ as the related $k$ nonzero elements of the reflectivity vector and  $\mathbf{A}_{\Omega_k}$ is a $N\times{k}$ matrix composed of the $k$ columns of matrix $\mathbf{A}$ whose indexes correspond to ${\Omega_k}$. The NLS estimate of $\Omega_k$ and  $\pmb{\gamma}_{\Omega_{k}}$ are as follows \cite{kay1993fundamentals}
\begin{equation}
\hat{\Omega}_{k}=arg\min_{\Omega_{k}\in{S}}\mathbf{g}^H\Pi^{\perp}_{\Omega_{k}}\mathbf{g}
\label{6}
\end{equation} 
\begin{equation}
\hat{\pmb{\gamma}}_{\Omega_{k}}=(\mathbf{A}^H_{\Omega_{k}}\mathbf{A}_{\Omega_{k}})^{-1}\mathbf{A}^H_{\Omega_{k}}\mathbf{g}.
\label{7}
\end{equation}
where $\Pi^{\perp}_{\Omega_{k}}=\mathbf{I}_N-\mathbf{A}_{\Omega_k}(\mathbf{A}^H_{\Omega_{k}}\mathbf{A}_{\Omega_k})^{-1}\mathbf{A}^H_{\Omega_{k}}$ is the orthogonal complement projection onto the subspace spanned by $\mathbf{A}_{\Omega_k}$. The combinatorial search in \eqref{6} is done in $S$ which is the possible location of targets obtained in previous step. This search space reduction makes accurate localization of targets feasible. For determining the number of targets, the model order slection is required which is done by optimizing a penalized likelihood criteria. Depending on noise variance being known or unknown, the model can be selected in two ways that is presented in the unified form as follows
\begin{equation}
\hat{k}=arg\min_{k}\{J_k\triangleq f(\varepsilon(k))+\mathcal{P}(k)\}
\label{8}
\end{equation} 
where $\varepsilon(k)$ and $f(x)$ are definde as
\begin{equation}
\varepsilon(k)=\min_{\Omega_{k}\in{S}}\{\mathbf{g}^H\Pi^{\perp}_{\Omega_{k}}\mathbf{g}\}.
\label{9}
\end{equation}
\begin{equation}
f(x)=\left\{
\begin{array}{ll}
\frac{x}{\sigma^2_n}, \hspace{1.89cm}  known\hspace{.2cm}\sigma^2_n\\
\\
N\ln(\frac{x}{N}), \hspace{.83cm} unknown\hspace{.2cm}\sigma^2_n.
\end{array}
\right.
\label{10}
\end{equation}
and $\mathcal{P}(k)$ is the penalty term. For the sake of brevity, the details on derivation of \eqref{6}-\eqref{10} are refered in Appendix A. The flowchart of the proposed method is shown in Fig. \ref{fig:total-algorithm}. 

There are various penalty terms, among them, Akaike information criterion (AIC) \cite{akaike1974new} and the Bayesian information criterion (BIC) \cite{schwarz1978estimating} are commonly used. AIC and BIC belong to the family of information-theoretic criteria with the following penalty term \cite{wax1985detection}
\begin{equation}
\mathcal{P}(k)=\eta.K
\label{11}
\end{equation}
where $\eta$ is the penalty factor and $K$ is the number of unknown parameters. For AIC and BIC, the values of the penalty factor are $\eta_{AIC}=1$ and $\eta_{BIC}=0.5\ln{N}$, respectively. Due to model selection bias of AIC for finite samples, its corrected version (AICc) is also suggested where 
\begin{figure*}
 \centering
 \includegraphics[width=0.5\linewidth]{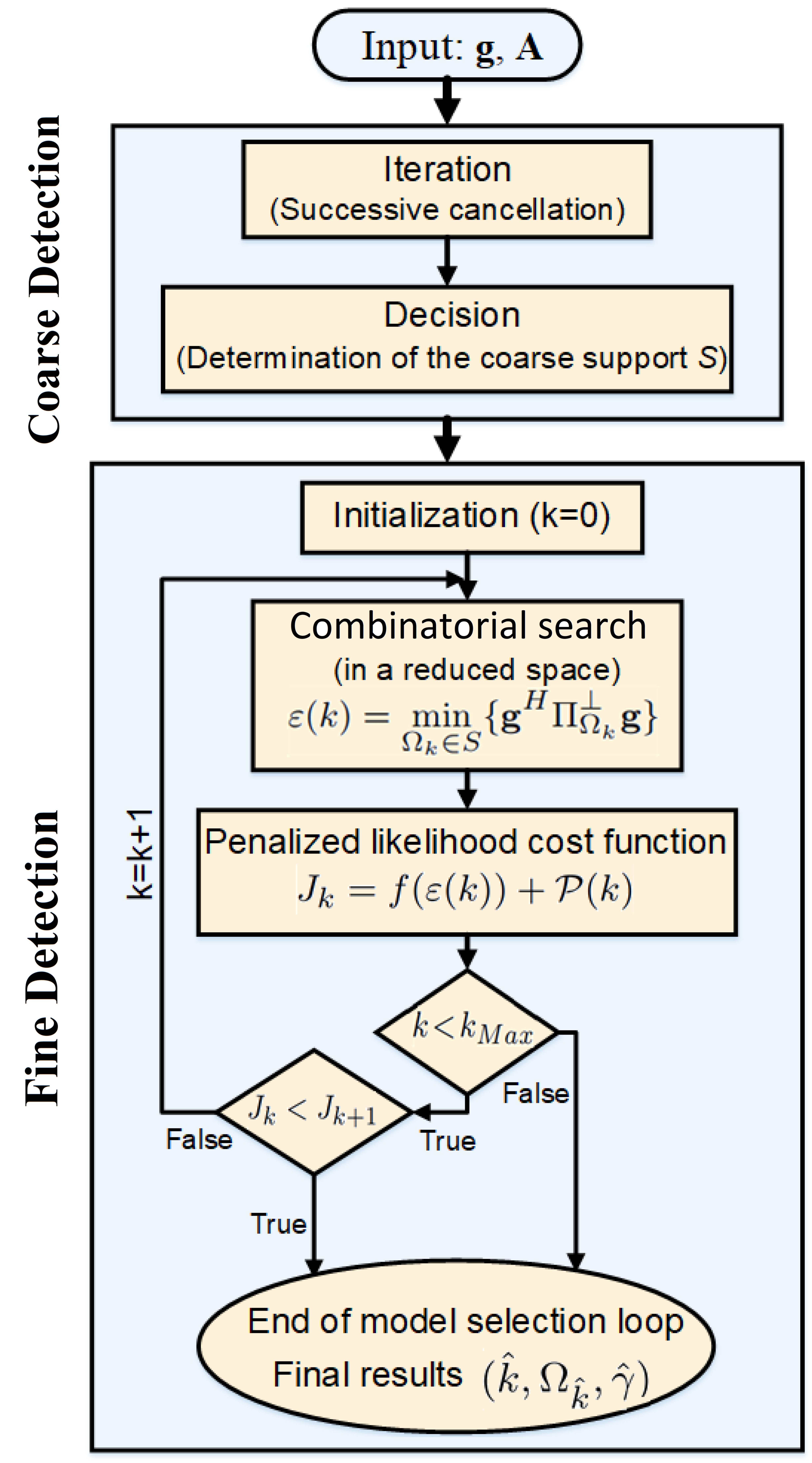}
 \caption{Flowchart of the proposed two-step detector (CA-NLS).}
 \label{fig:total-algorithm}
\end{figure*}
$\eta_{AICc}=\frac{N}{N-K-1}$ \cite{sugiura1978further}. Since there are three unknown parameters for every target (location, amplitude, and phase), the total number of unknown parameters is $K = 3k$. 

Overall, in the fine detection step, first, a combinatorial search is performed for each assumed $k$. Then, the model order is selected via minimization of the cost function \eqref{8} over $k \in [0,k_\mathrm{max}]$. However, if $J_k<J_{k+1}$, $k$ will be selected as the number of targets, and no further search is required for the higher number of targets.
\section{Performance criteria}\label{sec4}
In this section, the performance of the CA-NLS method in multiple scatterers' detection is analyzed. Then, through the simulation of system model \eqref{2}, the efficiency of CA-NLS is compared with the other methods including SGLRTC, SL1MMER, and Sup-GLRT. The first comparing method is a low-resolution method, while the second and third are high-resolution ones. The performance of the methods has been evaluated by two criteria: detection ability and estimation accuracy. 

In this section, we consider the presence of a maximum of two targets per a range-azimuth cell ($k_\mathrm{max}=2$). SNR per scatterer is defined as the square of the amplitude of a scatterer to the power of the noise. In double scatterers scenario, amplitudes of the two scatterers are considered the same and equal to $\sigma_s$. Thus, it is defined as $SNR=\frac{\sigma^2_s}{\sigma^2_n}$. In addition, the distance of targets is shown by $D_s$ and the normalized distance of targets $\alpha$ is defined as $\alpha=\frac{D_s}{\rho_s}$.
The parameters chosen for the simulation are the number of passes $N=20$, baseline extent $\Delta b=903$ m resulting in Rayleigh resolution $\rho_s=26$ m, elevation extent $\Delta_s=360$ m, number of grids $M=234$ or equivalently 17 points per $\rho_s$ and the threshold $T=0.8$ (to have probability of false alarm $P_{FA}=10^{-3}$ for SGLRTC method). 
\subsection{Detection ability}\label{sec4_1}
For the detection study, $\mathcal{H}_i$ is defined as the hypothesis of the presence of $i$ targets in the elevation direction and $\mathcal{D}_k$ as decision about presence of $k$ targets. The probability of detection for the resolution problem is defined as
\begin{equation}
P_{D} = P(\mathcal{D}_2|\mathcal{H}_2).
\label{12}
\end{equation}

Let's define $P_{D1}$ and $P_{D2}$ as the probabilities of detection of the coarse and fine detectors, respectively. As a cascaded detector, it can be shown that the total probability of detection will be the minimum of the two probabilities of detection. Hence for an appropriate threshold-setting, it is recommended to adjust $T$ in the first detector and penalty factor $\eta$ in the second detector to ensure $P_{D1} > P_{D2}$. Otherwise, missed detection in the coarse detection step may result in the null $S$ and the detection will not proceed to the fine step.

To analytically assess the performance of the proposed two-step method, $P_{D1}$ and $P_{D2}$ should be derived. For the sake of simplicity, we emphasize on the derivation of $P_{D2}$ and leave the exact derivation of $P_{D1}$ for future works. However, an upper-bound has been derived for $P_{D1}$ in \cite{pauciullo2012detection} which is valid only for faraway targets. In the rest of this section, to consider only the detection ability of the fine detector, the probability of detection of the coarse detector is assumed to be one i.e. $P_{D1}=1$ which yields $P_D=P_{D2}$. 

According to the definition of probability of detection in \eqref{12} and the model order selection procedure in \eqref{8}, the detection is successful when the cost function of the two-targets case is the least one, i.e. $J_{2} < J_{1}$. By using definition of $J_n$ in \eqref{8} for known noise variance case, we have
\begin{equation}
\frac{{\|\mathbf{g}-\mathbf{A}_{\Omega_{2}}\gamma_{\Omega_{2}}\|}^2}{\sigma_n^2}+\mathcal{P}(2) < \frac{{\|\mathbf{g}-\mathbf{A}_{\Omega_{1}}\gamma_{\Omega_{1}}\|}^2}{\sigma_n^2}+\mathcal{P}(1).
\label{13}
\end{equation}
where each side of \eqref{13} composed of a likelihood term and a penalty term. 

On the left-hand side of \eqref{13}, in the likelihood term, the location of two steering vectors is searched where $\mathbf{g}$ is composed of the noise plus two steering vectors. Under good SNR conditions, the estimated location of targets approximately coincides with their real value. Hence, after the cancellation of steering vectors of the estimated targets, the likelihood term on the left-hand side of \eqref{13} becomes the power of a pure Gaussian noise scaled to the noise variance. Thus, the likelihood term on the left-hand side of \eqref{13} will have a central chi-square distribution with $N$ degrees of freedom 
\begin{equation}
\frac{1}{\sigma_n^2}{\|\mathbf{g}-\mathbf{A}_{\Omega_{2}}\gamma_{\Omega_{2}}\|}^2\approx\frac{1}{\sigma_n^2}{\|\mathbf{n}\|}^2 \sim \chi_{N}^2.
\label{14}
\end{equation}

On the right hand side of \eqref{13}, in the likelihood term, $\mathbf{g}$ is composed of the noise plus two steering vectors which is subtracted from one steering vector related to the falsely detected target. Consequently, the likelihood term on the right hand side of \eqref{13} becomes the power of a pure Gaussian noise plus an additional term scaled to the noise variance. As a result, the likelihood term on right hand side of \eqref{13} will have a noncentral chi-square distribution with $N$ degrees of freedom
\begin{equation}
\frac{1}{\sigma_n^2}{\|\mathbf{g}-\mathbf{A}_{\Omega_{1}}\gamma_{\Omega_{1}}\|}^2\approx\frac{1}{\sigma_n^2}{\|\mathbf{n+r}\|}^2 \sim \chi_{N}^2(\lambda_r)
\label{15}
\end{equation}
where $\mathbf{r}$ is residual term added to Gaussian noise and $\lambda_r$ is noncentrality of the above chi-square distribution which are defined respectively as follows
\begin{equation}
\mathbf{r} = \sigma_s\boldsymbol{[}\boldsymbol{a}(s_{M_1}) + e^{j\Delta{\phi}}\boldsymbol{a}(s_{M_2}) - \mathbf{\hat{\gamma_p}}\boldsymbol{a}(s_p)\boldsymbol{]}
\label{16}
\end{equation}
\begin{equation}
\lambda_{r} =\frac{\|\mathbf{r}\|^2}{\sigma_n^2}
\label{17}
\end{equation} 
where $s_{M_1}$ and $s_{M_2}$ are real locations of targets, $\Delta{\phi}$ is their phase difference, $s_p$ is the estimated single-target location and $\mathbf{\hat{\gamma_p}}$ is the estimated reflectivity of the single-target scaled to $\sigma_s$. By rearranging \eqref{13} and using \eqref{14},\eqref{15}, we obtain
\begin{equation}
\frac{1}{\sigma_n^2}({{\|\mathbf{n+r}\|}^2}-{\|\mathbf{n}\|}^2) > \mathcal{P}(2) - \mathcal{P}(1).
\label{18}
\end{equation}

According to the ITC penalty in \eqref{11}, the above inequality can be rewritten as
\begin{equation}
\frac{1}{\sigma_n^2}({{\|\mathbf{n+r}\|}^2}-{\|\mathbf{n}\|}^2) > 3\eta.
\label{19}
\end{equation}

The following proposition shows that the left hand side of \eqref{19} has a Gaussian distribution.
\begin{prop}
 \label{proposition 1}
 : Let $X$ be a random variable with central chi-square distribution $\chi^2_N$ and $Y$ a random variable with noncentral chi-square distribution $\chi^2_N(\lambda_r)$ with $N$ degrees of freedom, both built by a common circular Gaussian random vector $\mathbf{n}$. The distribution of difference variable $Z=Y-X$ will be Gaussian with the mean and variance values equal to $\lambda_r$ and $2\lambda_r$, respectively.
 
\end{prop}
\begin{proof}
 See Appendix B.
\end{proof}

As a result, $P_{D}$ can be derived as
\begin{align}
P_{D} &= P(Z>3\eta) \notag\\
&=\frac{1}{\sqrt{2\pi\sigma^2_Z}}\int_{3\eta}^{\infty}exp(-\frac{(Z-m_Z)^2}{2\sigma^2_Z})dZ \notag\\
&=Q(\frac{3\eta-m_Z}{\sigma_Z}) \notag\\
&=Q(\frac{3\eta}{\sqrt{2\lambda_r}}-\sqrt{\frac{\lambda_r}{2}})
\label{21}
\end{align}
where $Q(x) = (1/\sqrt{2\pi})\int_{x}^{\infty}exp^{-t^2/2}dt$ is the $Q$-function. Also, the noncentrality term in \eqref{17} can be expressed in terms of other parameters as
\begin{equation}
\lambda_r = N. SNR.\vartheta(\alpha,\Delta{\phi})
\label{22}
\end{equation}           
where $\vartheta(.)$ is a function of normalized distance of two targets $\alpha$ and their phase difference $\Delta{\phi}$. Derivation of $s_p$, $\mathbf{\hat{\gamma_p}}$ and $\lambda_r$ is given in Appendix C.

\begin{figure}
 \centering
 \includegraphics[width=0.70\linewidth]{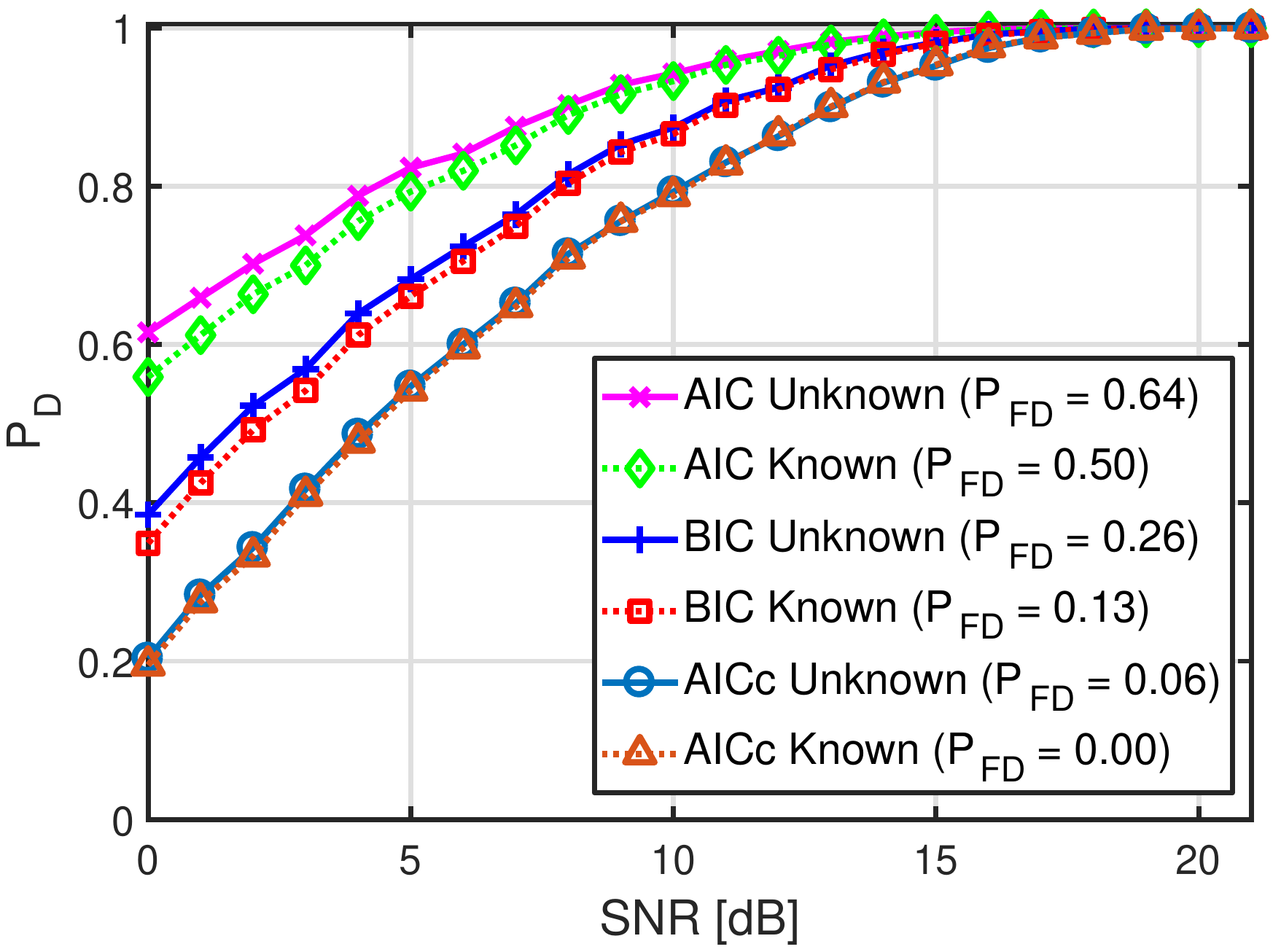}
 \caption{$P_D$ versus SNR for CA-NLS method respecting the penalty term used for model order selection. The normalized distance of targets is $\alpha=\frac{1}{2}$.}
 \label{fig:aicbicunknown}
\end{figure}

Due to the complexity of $\vartheta(.)$, this two-dimensional function is obtained by numerical calculation whose result is sketched in Fig. \ref{fig:vartheta}. It can be observed that for each value of $\Delta{\phi}$, the function $\vartheta(\alpha)$ is increasing with respect to $\alpha$ where $0<\alpha<1.34$. This range of $\alpha$ is reasonable as the relationships led to \eqref{22} are valid for closely-spaced targets. On the other hand, the effect of $\Delta{\phi}$ on $\vartheta(.)$ is also demonstrated in Fig. \ref{fig:vartheta}. By looking at the results, it is seen that the least probability of detection occurs in $\Delta{\phi}=0$.  

According to \eqref{22}, by increasing $N.SNR$ and $\alpha$, the noncentrality $\lambda_r$ increases. Since $Q(.)$ is a decreasing function with respect to its argument, increasing of $\lambda_r$ leads to increase of $P_{D}$. As a result, $N.SNR$ and $\alpha$ have increasing effects on $P_{D}$. The other influential parameter on $P_{D}$ is the penalty factor $\eta$. From \eqref{22}, it is evident that increase of $\eta$ leads to decrease of $P_{D}$. 

In the first numerical experiment, different cases of penalized cost function \eqref{8} have been compared. These cases are related to the availability of the information about the noise variance (known or unknown) and the penalty term used (BIC, AIC, and AICc). In Fig. \ref{fig:aicbicunknown}, $P_D$ of the different cases is sketched versus SNR. Also for these cases, the average probability of false detection $P_{FD}$ (i.e. $ P(\mathcal{D}_2|\mathcal{H}_1)$) is measured and has been indicated in the legend of the figure. 
\begin{figure}
 \centering
 \includegraphics[width=0.70\linewidth]{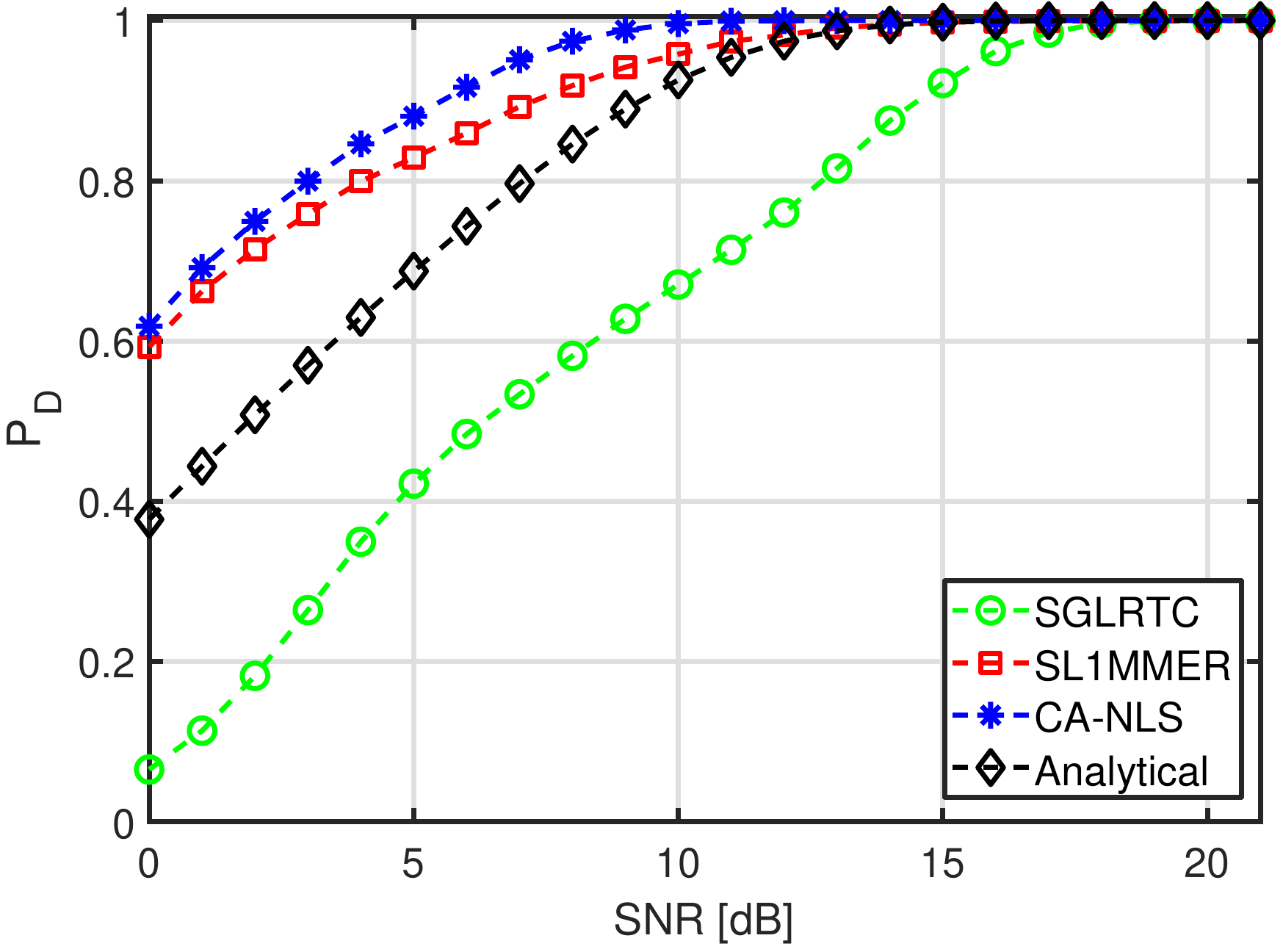}
 \caption{$P_D$ as a function of SNR for SGLRTC (green circles), SL1MMER (red squares), CA-NLS (blue stars) and analytical approximation (black diamonds), for $N=20$ and normalized distance of targets $\alpha=\frac{1}{2}$. }
 \label{fig:pdvssnr}
\end{figure} 

It is observed that the probability of detection in the case of unknown noise variance is slightly superior to the known noise variance one, but its average $P_{FD}$ is also higher. Besides, for AIC, $P_D$ and average $P_{FD}$ are higher comparing to BIC. For localization applications, when $P_{FD}$ increases, most of the single point scatterers are detected as two-point scatterers which makes the final three-dimensional SAR image different from reality. As a result, to maintain image quality, a trade-off between $P_D$ and $P_{FD}$ should be considered in designing the detector. 

For different methods, $P_D$ is shown versus SNR and $\alpha$ in Figs. \ref{fig:pdvssnr} and \ref{fig:pdvsspc}, respectively. Moreover, the analytical approximation of $P_D$ for CA-NLS (mentioned in \eqref{21}-\eqref{22}) has been plotted. For a fair comparison with the simulated methods using random $\Delta{\phi}$, the analytical function $\vartheta(\alpha,\Delta{\phi})$ is integrated over $\Delta{\phi}\in{[-\pi,\pi]}$ .To select the model order, the known noise variance case with BIC penalty term is used. Also, for the SL1MMER method, unconstraint form of Basis Pursuit DeNoising (BPDN) is implemented \cite{chen2001atomic} as follows
\begin{figure}
 \centering
 \includegraphics[width=0.70\linewidth]{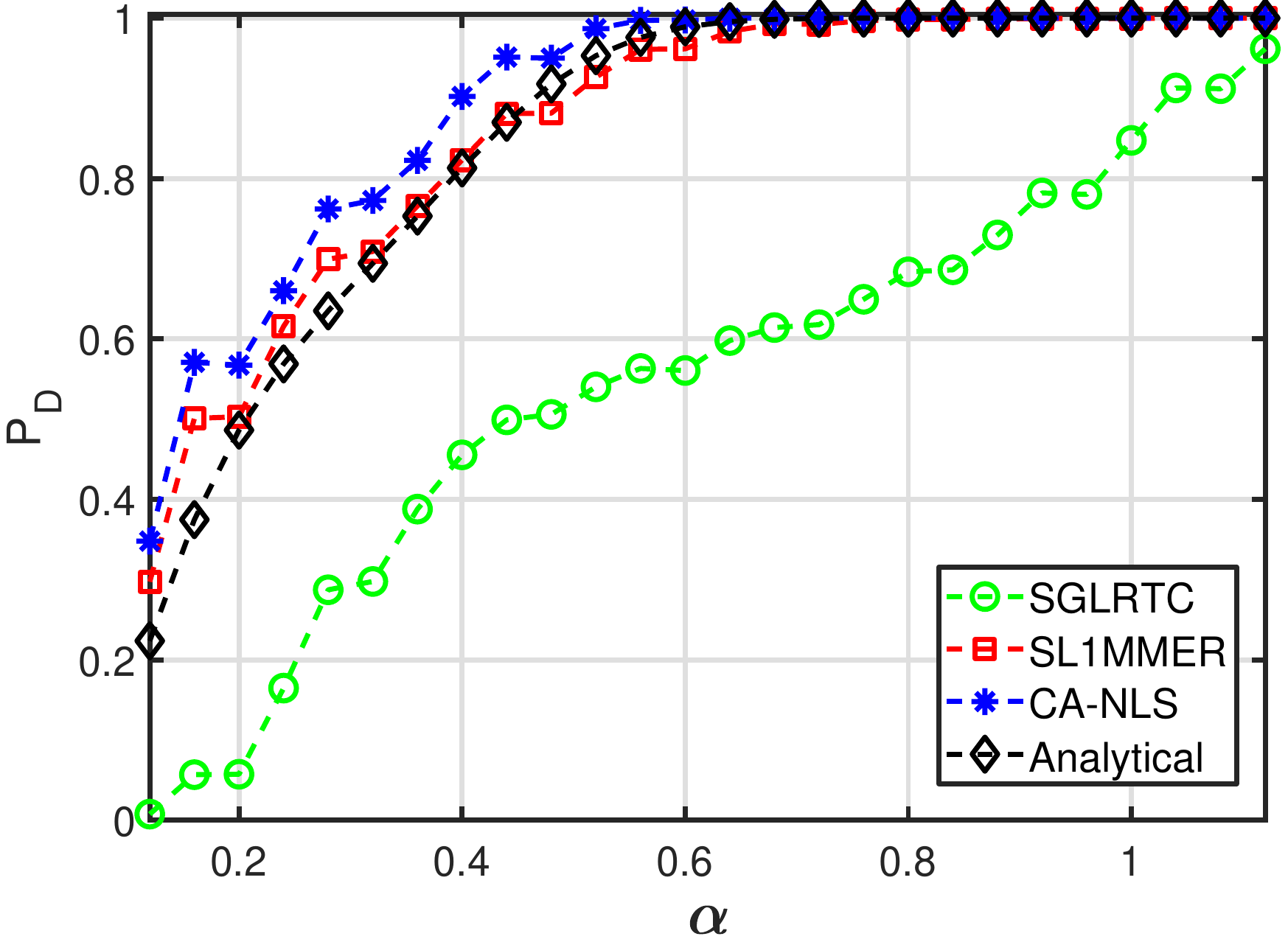}
 \caption{$P_D$ as a function of normalized distance of targets $\alpha$ for SGLRTC (green circles), SL1MMER (red squares), CA-NLS (blue stars) and analytical approximation (black diamonds), for $N=20$ and $SNR = 9\hspace{.1cm}dB$.}
 \label{fig:pdvsspc}
\end{figure}
\begin{equation}
\hat{\mathbf{\gamma}}=arg\min_{\mathbf{\gamma}}\{\|\mathbf{g}-\mathbf{A}\mathbf{\gamma}\|_2^2 + \lambda_{k}\|\gamma\|_1\} 
\label{23}
\end{equation}
where $\lambda_{k}$ is an adjusting parameter which trades off between reconstruction fidelity and sparsity.  
The convex minimization problem in \eqref{23} was solved by the Matlab CVX package \cite{grant2014cvx}. 

As expected, the results show that the probability of detection rises by increasing the SNR and distance of targets. It is observed that CA-NLS has the highest probability of detection, SL1MMER is close to but lower than the CA-NLS, and SGLRTC has the least $P_D$. Also, average $P_{FD}$ of the methods is obtained as 0.002, 0.03, and 0.05 for SGLRTC, CA-NLS, and SL1MMER respectively. These results show the advantage of CA-NLS to SL1MMER in having both higher $P_D$ and lower $P_{FD}$. Furtherly, when SNR is high enough ($SNR>9$ dB), the  analytical expression \eqref{21} is close to the numerical result of CA-NLS.  
\begin{figure}
 \centering
 \includegraphics[width=0.70\linewidth]{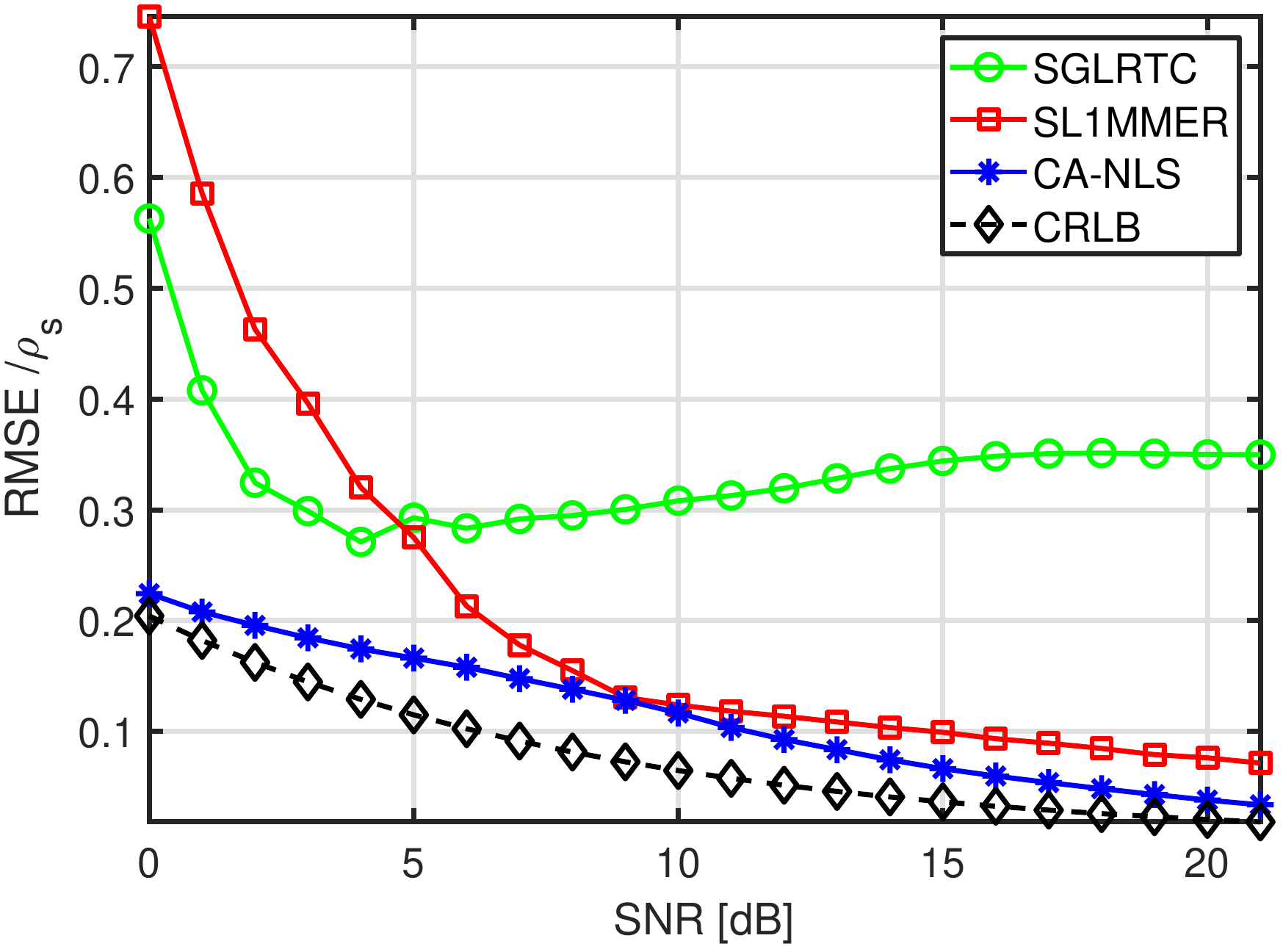}
 \caption{RMSE of estimated locations of targets versus $SNR$ for SGLRTC (green circles), SL1MMER (red squares), CA-NLS (blue stars) and square root of CRLB (black diamonds), for $N=20$ and normalized target's distance $\alpha=\frac{1}{2}$.}
 \label{fig:rmsevssnr}
\end{figure}
\subsection{Estimation accuracy}\label{sec4_2}
For evaluation of the estimation accuracy, the location of targets is considered where the reflectivity is not of concern for localization purpose. To have a criterion for the estimatin accuracy, Cramer Rao lower bound (CRLB) is used as an asymptotic bound for variance of an unbiased estimator. In SAR tomography, various studies have been conducted for obtaining CRLB \cite{de2009detection,zhu2012super}. Associated CRLB for the one-target and two-target cases are respectively as follows \cite{lee1992cramer,swingler1993frequency}
\begin{equation}
\hspace{1.1cm} CRLB_{1} = \frac{3}{2\pi^2}(\frac{\rho_s}{\sqrt{N.SNR}})^2, \hspace{1.03cm} k=1\\
\label{24}
\end{equation}
\begin{equation}
\hspace{1.1cm} CRLB_{2} \approx CRLB_{1}.\upzeta(\alpha,\Delta{\phi}), \hspace{0.95cm} k=2\\
\label{25}
\end{equation}  
where $\upzeta(.)$ is the normalized CRLB which is a function of $\alpha$ and $\Delta{\phi}$ with a value more than one such that $CRLB_2>CRLB_1$ \cite{swingler1993frequency,zhu2012super}. Also it is shown that the most estimation error occures at $\Delta{\phi}=0$ \cite{zhu2012super}. Since multiple snapshots are available in the real scenario, the phase difference of scatterers becomes random in nature. In this case, a $\Delta{\phi}$-independent expression is derived for normalized CRLB \cite{lee1992cramer,swingler1993frequency}
\begin{equation}
\upzeta(\alpha)\approx \max\{\frac{15}{\pi^2}\alpha^{-2},1\}
\label{26}
\end{equation} 
\begin{figure}
 \centering
 \includegraphics[width=0.70\linewidth]{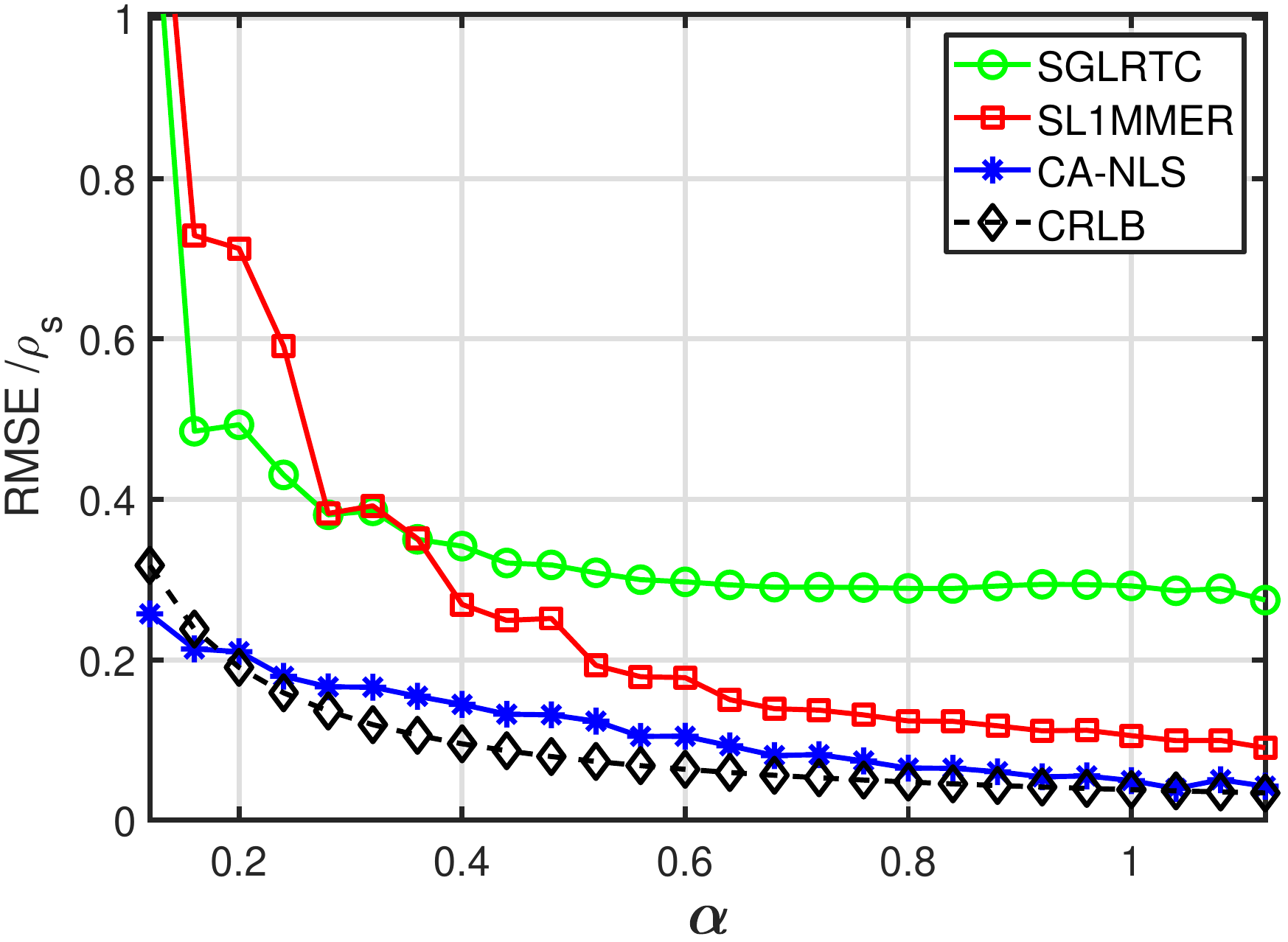}
 \caption{RMSE of estimated locations of targets versus normalized distance of targets $\alpha$ for SGLRTC (green circles), SL1MMER (red squares), CA-NLS (blue stars) and square root of CRLB (black diamonds), for $N=20$ and $SNR = 9\hspace{.1cm}dB$.}
 \label{fig:rmsevsspc}
\end{figure}

In Fig. \ref{fig:rmsevssnr}, RMSE of different methods is sketched versus SNR. To measure the estimation error, root mean square error (RMSE) of the location of point targets is considered
\begin{equation}
RMSE=\sqrt{\frac{1}{N_{Det}}\sum\limits_{i=1}^{N_{Det}}{\frac{(s_{M_1}-\hat{s}_{M_{1i}})^2+(s_{M_2}-\hat{s}_{M_{2i}})^2}{2}}}
\label{27}
\end{equation}
where $N_{Det}$ is the total number of the successful detections, $s_{M_1},s_{M_2}$ are real positions of targets in elevation and $\hat{s}_{M_{1i}},\hat{s}_{M_{2i}}$ are estimated positions of two targets at $i$-th successful detection. It should be noted that in the following results, RMSE of different methods together with the square root of CRLB is scaled to Rayleigh resolution $\rho_s$. It is seen that CA-NLS has the least estimation error and the error decreases by increasing SNR. For SGLRTC, it has a nearly constant error versus SNR which is higher than that of other methods. Similar to CA-NLS, the error of SL1MMER decreases by increasing SNR which is close to that of CA-NLS. However, the SL1MMER method has a high estimation error at lower SNR values. Hence, adequate SNR is required for using the SL1MMER method, while CA-NLS performs well even at low SNR conditions.

In the next experiment, RMSE is measured versus normalized distance of targets $\alpha$. As seen in Fig .\ref{fig:rmsevsspc}, the estimation error decreases by increasing the distance of targets. The performance of CA-NLS is superior to those of SL1MMER and SGLRTC. The error of CA-NLS and SL1MMER decreases by increasing $\alpha$ and approach zero, but for SGLRTC the error is nearly constant for $\alpha<1.5$.
\section{Experiments on simulated SAR data}\label{sec5}
In this section, the performance of different detection methods is measured using simulated SAR data. To do this, two types of experiments have been performed: layover separation and three-dimensional object reconstruction. Due to working in the real scenario, when the noise variance is not known, the unknown noise variance case is used in \eqref{8} for model selection.

\subsection{Layover separation}\label{subsec4_4_1}
In this experiment, two scatterers are considered, which are located on the same range-azimuth cell, one on the ground
and the other on the building facade. The aim is to measure the separability of the targets when their spacing varies from low to high. This type of evaluation has been previously used in \cite{zhu2010tomographic,schmitt2014maximum}. The simulation parameters are shown in Table \ref{tab:SimParam}. For preparing the tomographic stack, Range-Doppler Algorithm (RDA) has been used for focusing SAR images. Next, coregistration and deramping have been applied to the focused images.
\begin{table}[ht]
 \caption{SAR Simulation Parameters.} 
 \label{tab:SimParam}
 \begin{center}       
  \begin{tabular}{|l||l|} 
   \hline
   \rule[-1ex]{0pt}{3.5ex}  Parameter & Value \\
   \hline\hline
   \rule[-1ex]{0pt}{3.5ex}  Platform velocity\, (m/s) & 200 \\
   \hline
   \rule[-1ex]{0pt}{3.5ex} Carrier frequency\, (GHz) & 4.5 \\
   \hline
   \rule[-1ex]{0pt}{3.5ex} Pulse repetition frequency (PRF)\,(Hz) & 300 \\
   \hline
   \rule[-1ex]{0pt}{3.5ex} Range distance to center of scene (km)  & 20 \\
   \hline
   \rule[-1ex]{0pt}{3.5ex} Number of baselines (N) & 24 \\
   \hline
   \rule[-1ex]{0pt}{3.5ex} SNR (dB) & 9 \\
   \hline   
  \end{tabular}
 \end{center}
\end{table} 
\begin{figure}[ht]
 \centering
 \includegraphics[width=1\linewidth]{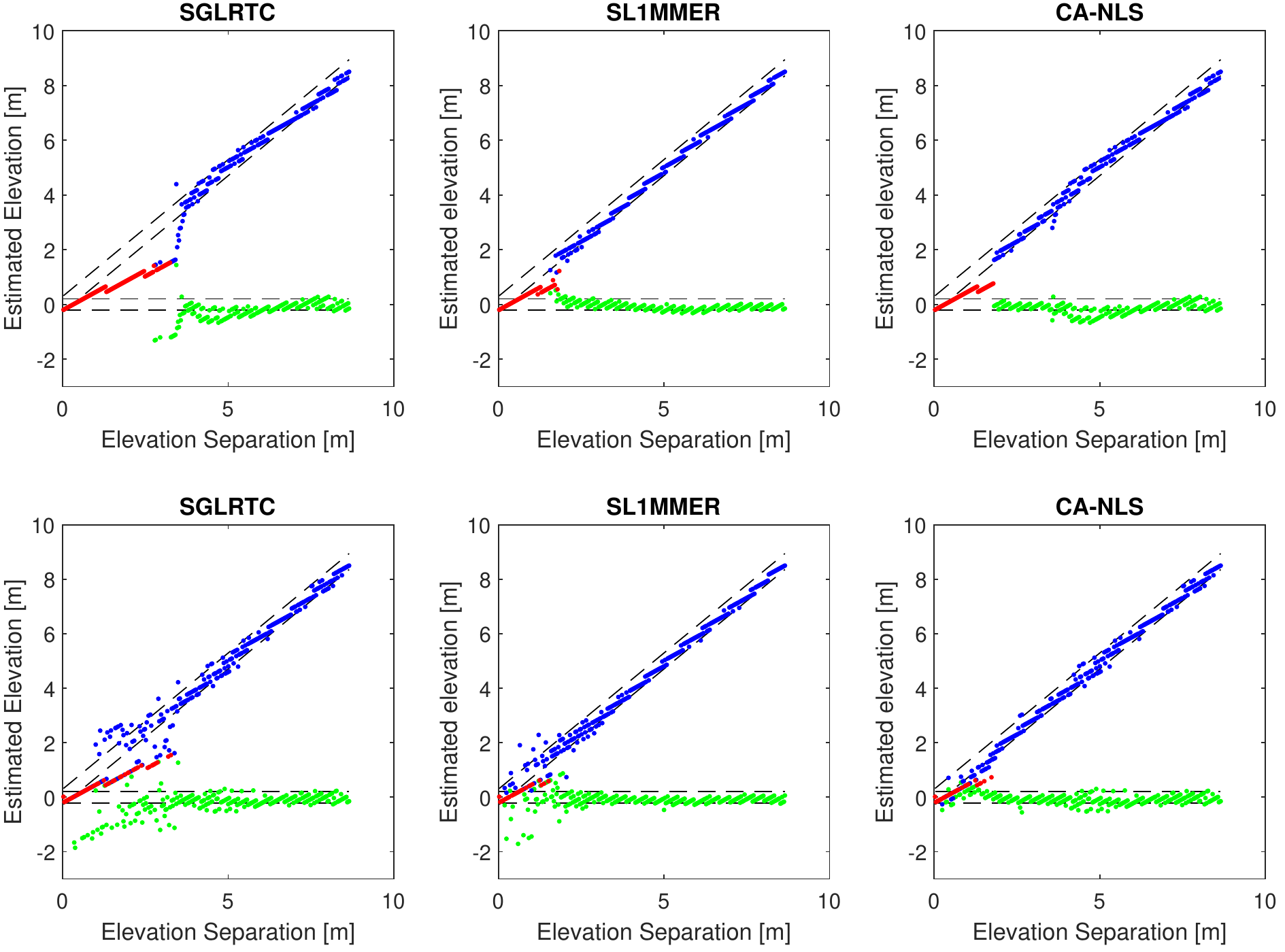}
 \caption{ Estimated elevations of double scatterers located on ground and facade at a linear increasing spacing. The sub-figures (a), (c) are for SGLRTC, (b), (d) are for SL1MMER and (c), (e) are for CA-NLS. The top row indicates targets with equal phase $\Delta{\phi}=0$ and the bottom row shows targets with unequal phases $\Delta{\phi}\neq0$. (Green and blue points indicate detected double scatterers and the red points indicate unresolved single scatterers).}
 \label{fig:sarsimulation}
\end{figure}

The simulation result is shown in Fig. \ref{fig:sarsimulation}. The horizontal axis shows targets on the ground, and the diagonal line belongs to the second targets on the facade. Two parallel dashed lines (margin lines) show $\pm3$ times of square root of $CRLB_1$. The top subfigures are results of experiments done under $\Delta{\phi} = 0$ condition, and the bottom ones are the result of the uniformly random phase difference ones. In this figure, the number of red points indicates the number of missed detections. The deviation of detected points (blue or green) from margin lines indicates the amount of estimation error. 
\begin{figure}[ht]
 \centering
 \includegraphics[width=.9\linewidth]{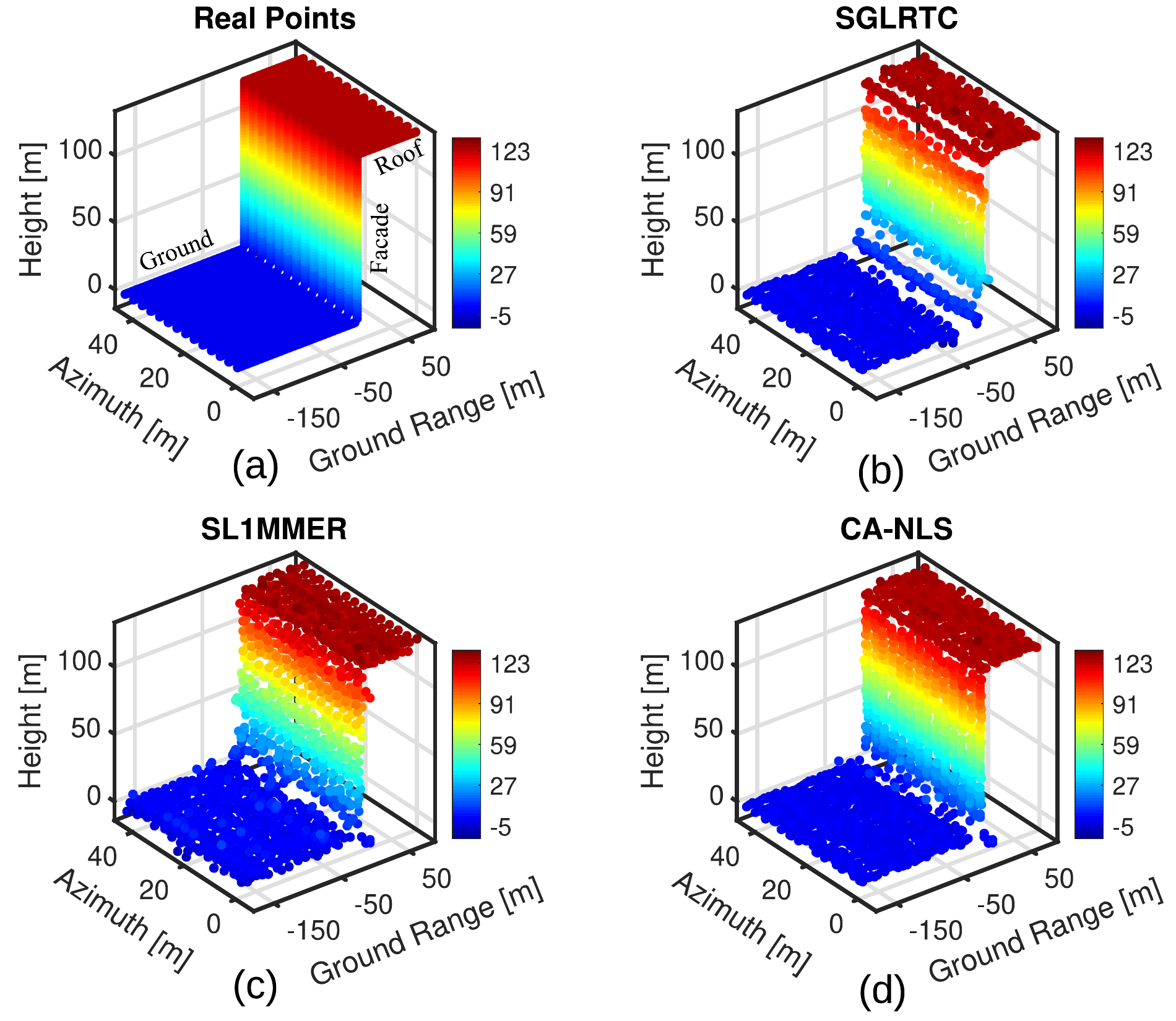}
 \caption{Three-dimensional reconstruction of a synthesized structure by tomographic inversion, (a) Real points (b) SGLRTC (c) SL1MMER and (d) CA-NLS.}
 \label{fig:synthesized3dtomo}
\end{figure}
\begin{figure}[ht]
 \centering
 \includegraphics[width=0.82\linewidth]{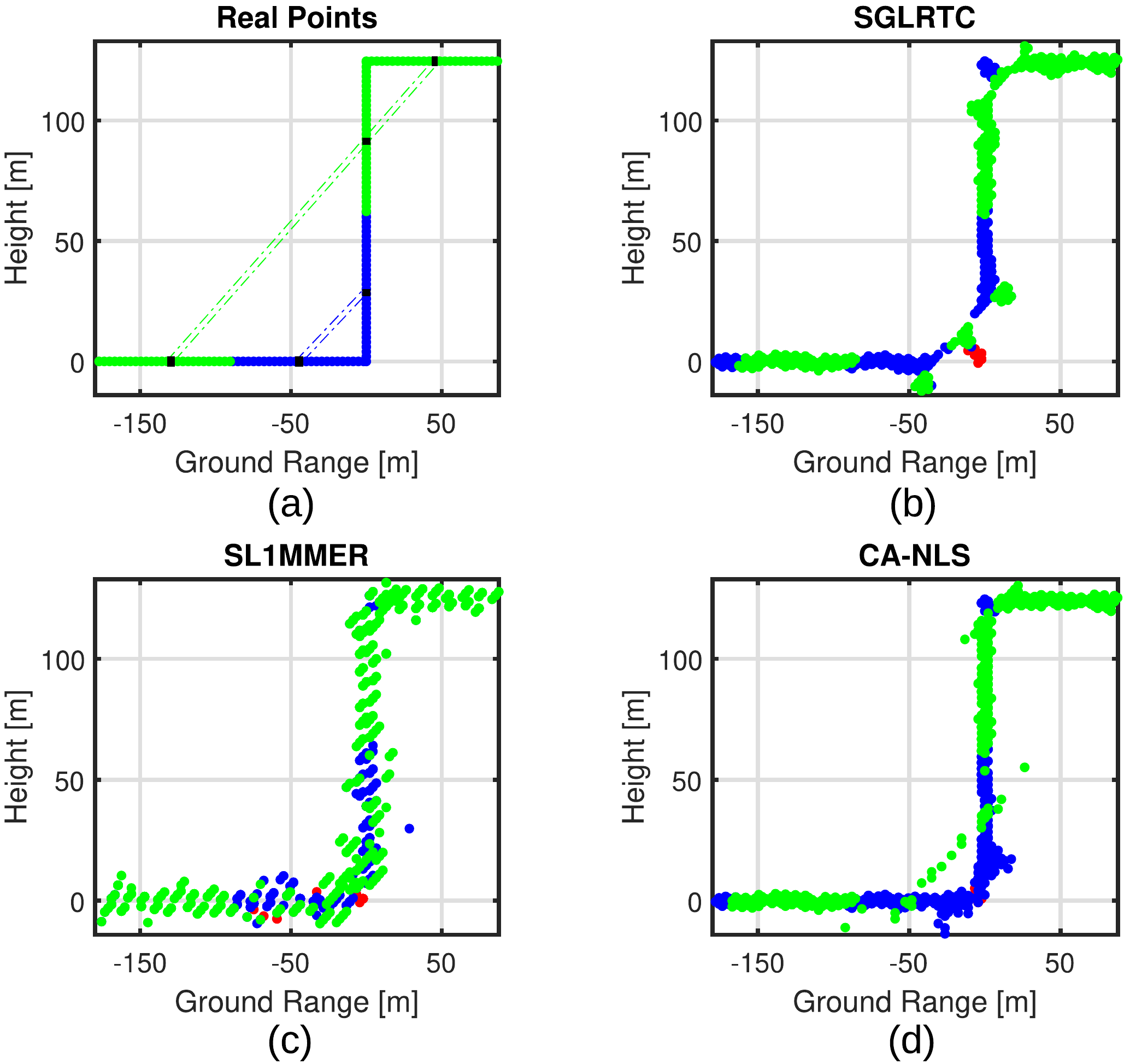}
 \caption{Side View of reconstructed images by TomoSAR methods,
  (a) Real points in which double scatterers and triple scatterers zones are indicated (b) SGLRTC (c) SL1MMER and (d) CA-NLS. Single, double and triple scatterer are indicated by red, blue, and green dots, respectively.}
 \label{fig:synthesized3dtomoside}
\end{figure}
The results of this experiment show the superiority of CA-NLS and SL1MMER methods as they can resolve targets even at the lower distance of targets  (the area close to the origin). Also, SGLRTC has the highest probability of missed detection and estimation error. By comparing the top and bottom sub-figures, we conclude that the non-zero phase difference increases the probability of detection (fewer red points) and decreases the estimation error. This verifies the role of $\Delta{\phi}$ in $CRLB_2$ \eqref{25} and the analytical expression of $P_D$ \eqref{21}-\eqref{22}. Finally, the experiment took $0.3$, $13.8$, and $524$ seconds for SGLRTC, CA-NLS, and SL1MMER, respectively, which shows a significant advantage of CA-NLS to SL1MMER in terms of complexity reduction.

\subsection{3D target reconstruction}\label{subsec4_4_1}
In this experiment, a three-part structure including facade, roof, and the ground is considered. The 3D structure is shown in Fig. \ref{fig:synthesized3dtomo} together with the results of TomoSAR reconstruction by SGLRTC, CA-NLS, and SL1MMER. The results show that CA-NLS performs better in reconstruction relative to the SGLRTC, especially at ground-facade and facade-roof conjunctions. In Fig. \ref{fig:synthesized3dtomoside}, side view images of reconstructed structure is shown. Although by SL1MMER the structure is formed, its parts are thickened and messy. In general, CA-NLS shows the highest performance qualitatively. 

To have a quantitative comparison, we analyzed the detection ability of the methods. For constructing ground, facade, and roof, a total of 1905 point targets were used, including 390 double scatterers (facade-ground) and 375 triple scatterers (roof-facade-ground) as indicated in Fig. \ref{fig:synthesized3dtomoside}. According to the results, it is observed that for the existing 390 double scatterers, CA-NLS, SL1MMER, and SGLRTC have detected 378, 244, and 326 double scatterers, respectively. These statistics show that the double scatterers are better detected by CA-NLS. Also, the total elapsed time of the experiment was $0.5$, $290$, and $420$ seconds for SGLRTC, CA-NLS, and SL1MMER, respectively, which shows the advantage of CA-NLS to SL1MMER in terms of computational cost.

Fig. \ref{fig:synthesized3dtomomos} shows the reconstruction results for CA-NLS regarding the penalty criteria used for model order selection. As the maximum number of targets is three $(k_\mathrm{Max}=3)$, the model selection scheme should compromise between the higher probability of detection and lower probability of false detection. As seen in Fig. \ref{fig:synthesized3dtomomos}, the quality of the AICc scheme is superior to that of AIC and BIC. The reason is the higher value of the penalty term for AICc compared with AIC and BIC which leads to less model overfitting. In other words, the probability of false detection, i.e. $P(\mathcal{D}_3|\mathcal{H}_2)$ for AICc is lower comparing to AIC and BIC. Hence the parts of the structure in which two scatterers are present at a range-azimuth pixel are better reconstructed by AICc, and consequently, a tidy image has resulted. According to the results, for the existing 390 double scatterers, BIC, AIC, and AICc have detected 278, 198, and 378 double targets, respectively. Also, for the 375 triple scatterers existed, BIC, AIC, and AICc have detected 474, 562, and 358 triple targets, respectively. It approves that AICc is the most suitable penalty criteria for the CA-NLS method.

\begin{figure}
 \centering
 \includegraphics[width=1\linewidth]{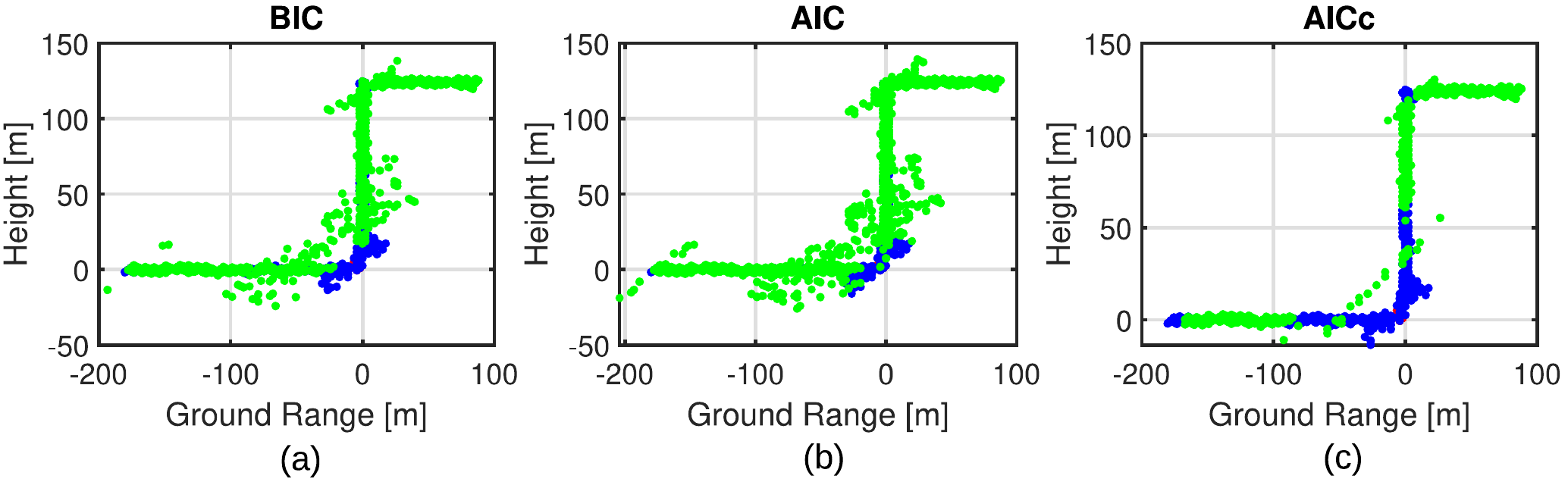}
 \caption{Comparison of model selection criteria used in CA-NLS for TomoSAR 3D reconstruction: (a) BIC (b) AIC and (c) AICc. }
 \label{fig:synthesized3dtomomos}
\end{figure}
Although the effect of deformation and thermal expansion has not been investigated in this paper, the performance of such exact modeling is predictable. By adding more dimensions, the layovered scatterers can be better resolved where the values of deformation velocity or thermal dilation of the two targets are likely different \cite{siddique2016single}. As a result, in most cases for 3D reconstruction, the fine detection step is not mandatory. The reason is the two close point targets obtained the coarse detection step are separable due to being located in different deformation or thermal bins. In \cite{budillon2017extension}, the authors have used a non-accurate fast version of Sup-GLRT that improves the detection performance by incorporating additional dimensions. However, some considerations should be given to the threshold setting in the multidimensional scenario.

\section{Computational complexity}\label{sec6}
A clear advantage of the CA-NLS is the reduction of the combinatorial search needed in the fine detection step. With the aid of coarse detection, the number of points (grids) in which, the $k$ targets should be searched, reduces from $M$ to $M_r$ and consequently, the number of likelihood function evaluations decreases from $M \choose k$ to $M_r \choose k$.   

In Table \ref{tab:Elapsed}, the elapsed time of different methods is shown for various numbers of grids. The normalized distance of targets is assumed $\alpha=\frac{1}{2}$. Simulations are performed in MATLAB R2017a environment, using Intel Core (TM) 2 Duo 7700k, 4.27 GHz. Choosing a denser grid does not affect the SGLRTC method but increases the time consumption of the Sup-GLRT method dramatically. For CA-NLS, using the coarse detection technique makes the computational complexity tolerable. Additionally, Sup-GLRT has another shortcoming that is the need for $k_\mathrm{max}$-dimensional search at each step of sequential detection, while for CA-NLS it is just necessary at the final step of model order selection. In other words, in the no-target case, the two-targets search (for finding $\Omega_{2}$) is not necessary for CA-NLS, while it is a must for Sup-GLRT.

For the closely spaced targets, the computational complexity of SL1MMER is more than that of CA-NLS. Also, for the SL1MMER method, applying compressed sensing is always obligatory even in the absence of any coherent target(s), where CA-NLS can easily reject the presence of target at coarse detection step.

\begin{table}[ht]
 \caption{Comparison of elapsed time (msec) for different methods.} 
 \label{tab:Elapsed}
 \begin{center}       
  \begin{tabular}{|l||l|l|l|l|} 
   \hline
   \rule[-1ex]{0pt}{3.5ex}  Number of grids & SGLRTC & CA-NLS & SL1MMER & Sup-GLRT  \\
   \hline\hline\rule[-1ex]{0pt}{3.5ex}  100 & 0.3 & 5 & 455 & 337  \\
   \hline
   \rule[-1ex]{0pt}{3.5ex}  200 & 0.3 & 9 & 467 & 1938 \\
            \hline
            \rule[-1ex]{0pt}{3.5ex}  300 & 0.3 & 16 & 484 & 3642\\
            \hline

  \end{tabular}
 \end{center}
\end{table} 
Beside search space reduction gained by CA-NLS, it has another advantage in complexity reduction. For the case when the targets are far apart, it can be shown that the fine detection step is not necessary. In other words, the location of targets estimated in the fine detection step is identical to the location of detected peaks acquired in the coarse detection step, as the following proposition states

\begin{prop}
 \label{proposition 2}
 Assume $supp_i, 0<{i\leq{k}}$ are the $k$ partial supports detected in the coarse detection step which are dis-contiguous i.e. $supp_i\bigcap{supp_j}=\emptyset,0<{i,j\leq{k}},i\neq{j}$. Then, the result of fine detection coincides with the associated peaks of detected partial supports in the coarse detection, i.e. $\Omega_{k}=\{s_{p_1},s_{p_2},...,s_{p_k}\}$.
\end{prop}

\begin{proof}
 See Appendix D.
\end{proof}

In other words the need for costly methods such as NLS and CS is eliminated when the targets are widely separated. It makes the proposed method computationally scalable regarding the spacing between the targets.
\section{Conclusion}\label{sec7}
In this paper, a high-resolution target detection method has been addressed for SAR tomography. The proposed two-step method utilizes the reduced search space derived from the correlation information in the first step for finding the accurate location of the targets in the second step. Also, an ITC based model order selection scheme has been used in the second step in two cases of known and unknown noise variance. The performance of the CA-NLS method was analyzed and investigated regarding detection ability, estimation accuracy, and computational cost. We have shown that our reduced-complexity method has the least estimation error and possesses the highest probability of detection with an acceptable probability of false detection. Additionally, we have demonstrated that the phase difference of the two targets is an effective parameter for detection ability other than $N$ and $SNR$. Furthermore, for far enough targets, we have shown that coarse detection suffices for target localization, and an accurate search is unnecessary. Overall, exploiting both CFAR detection and penalized likelihood criterion in the proposed two-step method gives a good compromise between $P_D$, $P_{FD}$, and the speed of detection. Further works can be focused on reaching a closed-form analytical solution for the case of closely-spaced targets, making the real-time detection of all kinds of targets possible, regardless of their spacing.

\appendix
\section[]{NLS estimation and model order selection}
\label{App_A}
\subsection[]{NLS estimation}
Assuming the additive noise in \eqref{2} is complex Gaussian, the likelihood function of the observed data in the presence of $k$ targets is defined as
\begin{equation}
p(\mathbf{g}|\mathbf{\pmb{\gamma}},k)=\frac{1}{\pi^N(\sigma_n^2)^N}exp(-\frac{{\|\mathbf{g}-\mathbf{A}_{\Omega_{k}}\pmb{\gamma}_{\Omega_{k}}\|}^2}{\sigma_n^2})
\label{28}
\end{equation}

The locations of the targets and their reflectivities are estimated by maximizing the log-likelihood function as follows
\begin{equation}
(\hat{\Omega}_{k},\hat{\pmb{\gamma}}_{\Omega_{k}})=arg\max_{\Omega_{k}}\{{\ln\,p(\mathbf{g}|\pmb{\gamma},k)}\}
\label{29}
\end{equation} 

Rewriting \eqref{29} using \eqref{28} gives
\begin{align}
(\hat{\Omega}_{k},\hat{\pmb{\gamma}}_{\Omega_{k}})&=arg\min_{\Omega_{k}}\{-\ln{p(\mathbf{g}|\pmb{\gamma},k)}\}\notag\\
&=arg\min_{\Omega_{k}}\{N\ln{2\pi}+N\ln{\sigma_n^2}+\frac{{\|\mathbf{g}-\mathbf{A}_{\Omega_{k}}\pmb{\gamma}_{\Omega_{k}}\|}^2}{\sigma_n^2}\}\notag\\
&=arg\min_{\Omega_{k}}\{{\|\mathbf{g}-\mathbf{A}_{\Omega_{k}}\pmb{\gamma}_{\Omega_{k}}\|}^2\}
\label{30}
\end{align} 
which, for a given $\Omega_{k}$, results in the least square estimate of reflectivity vector as \cite{kay1993fundamentals}
\begin{equation}
\hat{\pmb{\gamma}}_{\Omega_{k}}=(\mathbf{A}^H_{\Omega_{k}}\mathbf{A}_{\Omega_{k}})^{-1}\mathbf{A}^H_{\Omega_{k}}\mathbf{g}.
\label{31}
\end{equation} 

In order to find the location of targets, $\hat{\pmb{\gamma}}_{\Omega_{k}}$ is substituted in ${\|\mathbf{g}-\mathbf{A}_{\Omega_{k}}\pmb{\gamma}_{\Omega_{k}}\|}^2$ and $\hat{\Omega}_{k}$ is obtained as
\begin{equation}
\hat{\Omega}_{k}=arg\min_{\Omega_{k}}\mathbf{g}^H\Pi^{\perp}_{\Omega_{k}}\mathbf{g}
\label{32}
\end{equation} 
\subsection[]{Model order selection: known and unknown $\sigma_n^2$}
To find the number of targets, the following penalized likelihood criterion should be optimized 
\begin{equation}
\hat{k}=arg\min_{k}\{-\ln{p(\mathbf{g}|\hat{\pmb{\gamma}}_{\Omega_{k}},k)}+\mathcal{P}(k)\}
\label{33}
\end{equation}
where $\mathcal{P}(k)$ is the penalty term. Substituting the likelihood function \eqref{28} into \eqref{33} gives
\begin{equation}
\hat{k}=arg\min_{k}\{N\ln{2\pi}+N\ln{\sigma_n^2}+\frac{{\|\mathbf{g}-\mathbf{A}_{\Omega_{k}}\hat{\pmb{\gamma}}_{\Omega_{k}}\|}^2}{\sigma_n^2}+\mathcal{P}(k)\}
\label{34}
\end{equation}

The solution of the above problem depends on the information about the noise variance $\sigma_n^2$. If the noise variance is known, the number of targets is obtained by solving the following problem
\begin{equation}
\hat{k}=arg\min_{k}\{\frac{{\|\mathbf{g}-\mathbf{A}_{\Omega_{k}}\hat{\pmb{\gamma}}_{\Omega_{k}}\|}^2}{\sigma_n^2}+\mathcal{P}(k)\},\hspace{.8cm}  (Known\hspace{.15cm}\sigma^2_n).
\label{35}
\end{equation}

When the noise variance is unknown, it should be estimated from the received data. Using ML estimation we have
\begin{align}
\hat{\sigma}_n^2&=arg\min_{\sigma_n^2}\{-\ln{p(\mathbf{g}|\sigma_n^2,k)}\}\notag\\
&=arg\min_{\sigma_n^2}\{N\ln{2\pi}+N\ln{\sigma_n^2}+\frac{{\|\mathbf{g}-\mathbf{A}_{\Omega_{k}}\hat{\pmb{\gamma}}_{\Omega_{k}}\|}^2}{\sigma_n^2}\}
\label{36}
\end{align} 
The minimization problem \eqref{36} can be solved analytically by taking derivative of its argument with respect to $\sigma_n^2$ and setting it to zero. The estimation result is as follows \cite{kay2001conditional,stoica2004model}
\begin{equation}
\hat{\sigma}_n^2=\frac{1}{N}\|{\mathbf{g}-\mathbf{A}_{\Omega_{k}}{\hat{\pmb{\gamma}}_{\Omega_{k}}}}\|^2.
\label{37}
\end{equation}

By substituting $\hat{\sigma}_n^2$ in \eqref{34}, the third term of log-likelihood function becomes constant, equal to $N$. Hence, after omitting the constant terms of the likelihood function \eqref{34}, the number of targets is derived as
\begin{equation}
\hat{k}=arg\min_{k}\{N\ln(\frac{\|{\mathbf{g}-\mathbf{A}_{\Omega_{k}}}\hat{\pmb{\gamma}}_{\Omega_{k}}\|^2}{N})+\mathcal{P}(k)\},\hspace{.35cm} (Unknown\hspace{.15cm}\sigma^2_n).
\label{38}
\end{equation}

On the other side, at the NLS estimate of $\hat{\Omega}_{k}$ and $\hat{\pmb{\gamma}}_{\Omega_{k}}$ the  $\|{\mathbf{g}-\mathbf{A}_{\Omega_{k}}}\hat{\pmb{\gamma}}_{\Omega_{k}}\|^2$ term in \eqref{35}, \eqref{38} is minimized which is equal to
\begin{equation}
\varepsilon(k)=\min_{\Omega_{k}}\{\mathbf{g}^H\Pi^{\perp}_{\Omega_{k}}\mathbf{g}\}.
\label{39}
\end{equation}
By considering \eqref{35}, \eqref{38} and \eqref{39}, proof of \eqref{8} is complete.
\section{Proof of Proposition \ref{proposition 1}}
\label{App_B}
First, assume $\mathbf{n}$ is a zero-mean complex circular Gaussian random vector $\mathcal{CN}(0,\sigma^2_n\mathbf{I}_N)$, and $\mathbf{r}$ is a $N \times 1$ complex vector. Next, let's define $X$ and $Y $ as follows
\begin{align}
& X = \frac{1}{\sigma_n^2}\|\mathbf{n}\|^2 \notag\\
& Y = \frac{1}{\sigma_n^2}\|\mathbf{n+r}\|^2
\label{40}
\end{align}
which have standard central and noncentral chi-square distributions with $N$  degrees of freedom respectively. It is desired to obtain distribution of the difference between above random variables 
\begin{equation}
Z=\frac{1}{\sigma_n^2}(\|\mathbf{n+r}\|^2-\|\mathbf{n}\|^2).
\label{41}
\end{equation}
and also to derive its statistical parameters.

By rewriting the above equation, we have
\begin{align}
Z&=\frac{1}{\sigma_n^2}((\mathbf{n+r})^H(\mathbf{n+r})-\mathbf{n}^H\mathbf{n}) \notag\\
&=\frac{1}{\sigma_n^2}(\mathbf{n}^H\mathbf{n}+\mathbf{n}^H\mathbf{r}+\mathbf{r}^H\mathbf{n}+\mathbf{r}^H\mathbf{r}-\mathbf{n}^H\mathbf{n}) \notag\\
&=\frac{1}{\sigma_n^2}(\mathbf{n}^H\mathbf{r}+\mathbf{r}^H\mathbf{n})+\frac{\mathbf{r}^H\mathbf{r}}{\sigma_n^2} \notag\\
&=\frac{2}{\sigma_n^2}\mathrm{Re}(\mathbf{r}^H\mathbf{n})+\frac{\|\mathbf{r}\|^2}{\sigma_n^2}.
\label{42}
\end{align}
where $\mathrm{Re(.)}$ stands for real part.
According to \eqref{17} for the expression of the noncentrality term, we have
\begin{equation}
Z=\frac{2}{\sigma_n^2}\mathrm{Re}(\mathbf{r}^H\mathbf{n})+\lambda_r.
\label{43}
\end{equation}

By expanding the real term of $\mathbf{r}^H\mathbf{n}$, we have
\begin{align}
Z&=\frac{2}{\sigma_n^2}\boldsymbol{[}\mathrm{Re}(\mathbf{r^H})\mathrm{Re}(\mathbf{n})-\mathrm{Im}(\mathbf{r^H})\mathrm{Im}(\mathbf{n})\boldsymbol{]}+\lambda_r \notag\\
&=\frac{2}{\sigma_n^2}\boldsymbol{[}\mathrm{Re}(\mathbf{r})\mathrm{Re}(\mathbf{n})+\mathrm{Im}(\mathbf{r})\mathrm{Im}(\mathbf{n})\boldsymbol{]}+\lambda_r
\label{44}
\end{align}
where $\mathrm{Im}(.)$ stands for imaginary part.

Each of the two terms inside the bracket are linear combinations of zero mean Gaussian random variables $\mathrm{Re}(\mathbf{n})$ and $\mathrm{Im}(\mathbf{n})$, hence their summation has also zero mean Gaussian distribution. As a result, $Z$ has normal distribution with the mean equal to the noncentrality term 
\begin{equation}
m_Z = \lambda_r
\label{45}
\end{equation}
and the expression of the variance of $Z$ is
\begin{align}
\sigma^2_Z&=E\boldsymbol{[}(Z-m_Z)^2\boldsymbol{]} \notag\\
&=E\{\frac{2}{\sigma_n^2}\boldsymbol{[}\mathrm{Re}(\mathbf{r})\mathrm{Re}(\mathbf{n})+\mathrm{Im}(\mathbf{r})\mathrm{Im}(\mathbf{n})\boldsymbol{]}^2\} \notag\\
&=\frac{4}{\sigma_n^4}\{E\boldsymbol{[}(\mathrm{Re}(\mathbf{r})\mathrm{Re}(\mathbf{n}))^2\boldsymbol{]}+E\boldsymbol{[}(\mathrm{Im}(\mathbf{r})\mathrm{Im}(\mathbf{n}))^2\boldsymbol{]} \notag\\
&+2E\boldsymbol{[}\mathrm{Re}(\mathbf{r})\mathrm{Re}(\mathbf{n})\mathrm{Im}(\mathbf{r})\mathrm{Im}(\mathbf{n})\boldsymbol{]}\}.
\label{46}
\end{align}

Considering $\mathrm{Re}(\mathbf{n})$ and $\mathrm{Im}(\mathbf{n})$ that are two uncorrelated Gausian random variables, cross-correlation of every linear combination of them equals to zero. Consequently the third term of the \eqref{46} becomes zero. So
\begin{equation}
\sigma^2_Z=\frac{4}{\sigma_n^4}\{E\boldsymbol{[}(\mathrm{Re}(\mathbf{r})\mathrm{Re}(\mathbf{n}))^2\boldsymbol{]}+E\boldsymbol{[}(\mathrm{Im}(\mathbf{r})\mathrm{Im}(\mathbf{n}))^2\boldsymbol{]}\}.\\
\label{47}
\end{equation}

To calculate the remaining terms inside the brace, it is necessary to derive variance of a linear combination of $N$ sample zero mean Gausian random variables. Consider $\mathbf{x} = [x_1,x_2,...,x_N]^\mathrm{T}$ a vector of $N$ independent zero-mean Gausian random variables ($E(x_i^2)=\sigma_x^2$, $E(x_ix_j)=0, i\neq{j}$) and the weighting vector $\boldsymbol{a}=[a_1,a_2,...,a_N]$. Then we have
\begin{align}
E{(\boldsymbol{a}\mathbf{x})^2}&=E({\sum^N_{i=1}}a^2_ix^2_i+\sum^N_{i=1}\sum^N_{j=1,j\neq{i}}a_ix_ia_jx_j) \notag\\
&=\sum^N_{i=1}a^2_iE(x^2_i)+\sum^N_{i=1}\sum^N_{j=1,j\neq{i}}a_ia_jE(x_ix_j) \notag\\
&=\sigma_x^2\sum^N_{i=1}a_i^2 \notag\\
&=\sigma_x^2\|\boldsymbol{a}\|^2.
\label{48}
\end{align}

According to the above result, the variance of $Z$ can be derived from \eqref{48} as
\begin{align}
\sigma^2_Z&=\frac{4}{\sigma_n^4}(\sigma_{nR}^2\|\mathrm{Re}(\mathbf{r})\|^2+\sigma_{nI}^2\|\mathrm{Im}(\mathbf{r})\|^2) \notag\\
&=\frac{4}{\sigma_n^4}(\frac{\sigma_{n}^2}{2}\|\mathrm{Re}(\mathbf{r})\|^2+\frac{\sigma_{n}^2}{2}\|\mathrm{Im}(\mathbf{r})\|^2) \notag\\
&=\frac{2}{\sigma_n^2}\|\mathbf{r}\|^2 \notag\\
&=2\lambda_r
\label{49}
\end{align}
where $\sigma_{nR}^2$ and $\sigma_{nI}^2$ are variance of real and imaginary parts of the random vector $\mathbf{n}$ which because of their independence, are equal to half the total variance $\sigma_{n}^2$. 

Therefore, $Z$ is a Gaussian random variable with mean equal to the non-centrality value and variance equal to two times of that value, and the proof is complete.

\section{Derivation of $s_p$, $\hat{\boldmath{\gamma_p}}$ and $\lambda_r$}
\label{App_C}
To derive the expression of $\lambda_r$, the parameters $s_p$ and $\hat{\boldmath{\gamma_p}}$ should be available. Initially, it is desired to find the location of the falsely-detected target $s_p$ where two targets are present at $s_{M_1}$ and $s_{M_2}$. According to \eqref{6}, $s_p$ is obtained by minimizing the following cost function
\begin{equation}
\mathbf{J}={\mathbf{g}^{H}\boldsymbol{[}\mathbf{I}_N-\mathbf{A}_{\Omega_1}(\mathbf{{A^H}}_{\Omega_1}\mathbf{A}_{\Omega_1})^{-1}\mathbf{{A^H}}_{\Omega_1}\boldsymbol{]}\mathbf{g}}
\label{50}
\end{equation} 
As ${\mathbf{g}}^{H}\mathbf{g}$ is a constant value, the minimization in \eqref{50} is equivalent to maximization of the following cost function
\begin{equation}
\mathbf{J}=\mathbf{g}^{H}\mathbf{A}_{\Omega_1}(\mathbf{{A^H}}_{\Omega_1}\mathbf{A}_{\Omega_1})^{-1}\mathbf{{A^H}}_{\Omega_1}\mathbf{g}.
\label{51}
\end{equation} 
As $s_p$ is the maximizer the above cost function, $\mathbf{A}_{\Omega_1}=\boldsymbol{[}\mathbf{\textit{a}}(s_p)\boldsymbol{]}$. Since steering vectors are coloumns of a Fourier matrix, we have $\boldsymbol{a}^H(s_p)\boldsymbol{a}(s_p) = N$. Hence
\begin{equation}
\mathbf{J}={\frac{1}{N}}{\boldsymbol{|}\boldsymbol{a}^H(s_p)\mathbf{g}\boldsymbol{|}}^2.
\label{52}
\end{equation} 

Neglecting the noise term, $\mathbf{g}$ is summation of two steering vectors at $s_{M_1}$ and $s_{M_2}$ multiplied by $\sigma_s$. Hence $s_p$ is the optimizer of following cost function
\begin{equation}
\mathbf{J}={\frac{\sigma_s^2}{N}}{\boldsymbol{|}\boldsymbol{a}^H(s_p)(\boldsymbol{a}(s_{M_1}) + e^{j\Delta{\phi}}\boldsymbol{a}(s_{M_2}))\boldsymbol{|}}^2.
\label{53}
\end{equation} 

By introducing the correlation between $m$th and $n$th steering vectors  
\begin{equation}
\mathcal{C}_{m,n} = \frac{\boldsymbol{a}^H(s_{m})\boldsymbol{a}(s_{n})}{N}
\label{54}
\end{equation}
the cost function \eqref{53} can be rewritten as
\begin{equation}
\mathbf{J}={N\sigma_s^2}{\boldsymbol{|}\mathcal{C}_{p,M_1}+e^{j\Delta{\phi}}\mathcal{C}_{p,M_2}\boldsymbol{|}}^2.
\label{55}
\end{equation} 

According to the definition of steering vectors, the value of correlation can be obtained as follows

\begin{equation}
\mathcal{C}_{m,n}={\frac{1}{N}}\sum^N_{i=1}exp(\frac{j4\pi(s_n-s_m)b_i}{\lambda{R_0}}).
\label{56} 
\end{equation}

Due to equi-spaced sampling in the baseline and elevation directions, we can write $b_i=(\frac{i}{N-1})\Delta_b$ and $s_n-s_m=\frac{(n-m)}{M-1}\Delta_s$. Then, by applying some mathematical operations (summation of a geometric progression) on \eqref{56}, the correlation term can be expressed as the following digital sinc function
\begin{equation}
\mathcal{C}_{m,n}=\exp(-jL(n-m))\frac{\sin(NL(n-m))}{N\sin(L(n-m))}
\label{57}
\end{equation}
where the factor $L$ is equal to
\begin{equation}
L=\frac{2{\pi}\Delta s\Delta b}{(M-1)(N-1)\lambda{R_0}}.
\label{58}
\end{equation}

As is seen in \eqref{57}, $\mathcal{C}_{m,n}$ is a function of the difference of indexes of two steering vectors, and can be used as $\mathcal{C}_{mn}=\mathcal{C}(n-m)$. Also conjugate symmetry property is held, i.e. $\mathcal{C}(n-m)=\mathcal{C}^*(m-n)$. 

For the correlation between $M_1$ and $M_2$ indexes related to the two targets, $\mathcal{C}_{m,n}$ in \eqref{57} can be more simplified. Due to the definition of the normalized distance of targets $\alpha=(s_{M_2}-s_{M_1})/\rho_s$ and Rayleigh resolution in \eqref{3}, the following expression is resulted
\begin{equation}
\mathcal{C}(M_2-M_1)=\exp(-jL'\alpha)\frac{\sin(NL'\alpha)}{N\sin(L'\alpha)}
\label{59}
\end{equation}
where the factor $L'=\pi/(N-1)$ and is obtained with regard to the definition of $\rho_s$ and $\alpha$.

Afterwards, by substituting \eqref{59} in \eqref{55} and through some inequality math operations, it can be shown that $s_p$ the position of optimizer of \eqref{55} is the average of $s_{M1}$ and $s_{M2}$
\begin{equation}
s_p = \frac{s_{M_1}+s_{M_2}}{2}.
\label{60}
\end{equation}

Equivalently, due to correspondence between elevation samples and their indexes $p = (M_1+M_2)/2$. Now it is sufficient to get the least square estimated value $\hat{\boldmath{\gamma_p}}$ at location $(s_{M_1}+s_{M_2})/2$ with regard to \eqref{7} (after scaling to $\sigma_s$)
\begin{align}
\hat{\boldmath{\gamma_p}}&=\frac{1}{\sigma_s}.(a^H(s_p)a(s_p))^{-1}a^H(s_p)\mathbf{g} \notag\\
&=\frac{1}{N}a^H(s_p)\frac{\mathbf{g}}{\sigma_s} \notag\\
&=\frac{1}{N}a(\frac{s_{M_1}+s_{M_2}}{2})^H(a(s_{M_1})+e^{j\Delta{\phi}}a(s_{M_2})) \notag\\
&=\mathcal{C}(M_1-\frac{M_1+M_2}{2})+e^{j\Delta{\phi}}\mathcal{C}(M_2-\frac{M_1+M_2}{2}) \notag\\
&=\mathcal{C}(\frac{M_1-M_2}{2})+e^{j\Delta{\phi}}\mathcal{C}(\frac{M_2-M_1}{2}) \notag\\
\label{61}
\end{align}
Regarding \eqref{16} and \eqref{17}, the noncentrality term can be derived as follows
\begin{align}
\lambda_r =&\frac{\|\mathbf{r}\|^2}{\sigma_n^2} \notag\\
=& \frac{\sigma_s^2}{\sigma_n^2}\|\boldsymbol{a}(s_{M_1}) + e^{j\Delta{\phi}}\boldsymbol{a}(s_{M_2}) - \hat{\boldmath{\gamma_p}}\boldsymbol{a}(s_p)\|^2 \notag\\
= & SNR\{\|\boldsymbol{a}(s_{M_1}))\|^2 + \|\boldsymbol{a}(s_{M_2}))\|^2 + |\hat{\boldmath{\gamma_p}}|^2 \|\boldsymbol{a}(s_p))\|^2 \notag\\
& + 2N.\mathrm{Re}\boldsymbol{[}e^{j\Delta{\phi}}\mathcal{C}(M_2-M_1) - \hat{\boldmath{\gamma_p}}.\mathcal{C}(p-M_1) \notag\\
& - e^{-j\Delta{\phi}}\hat{\boldmath{\gamma_p}}.\mathcal{C}(p-M_2)\boldsymbol{]}\} \notag\\
= & SNR\{2N+|\hat{\boldmath{\gamma_p}}|^2N+2N.\mathrm{Re}\boldsymbol{[}e^{j\Delta{\phi}}\mathcal{C}(M_2-M_1) \notag\\
& - \hat{\boldmath{\gamma_p}}.(\mathcal{C}(\frac{M_2-M_1}{2}) + e^{-j\Delta{\phi}}\mathcal{C}(\frac{M_1-M_2}{2})) \boldsymbol{]}\} \notag\\
= & N.SNR. \{2+|\hat{\boldmath{\gamma_p}}|^2+2.\mathrm{Re}\boldsymbol{[}e^{j\Delta{\phi}}\mathcal{C}(M_2-M_1) -|\hat{\boldmath{\gamma_p}}|^2\boldsymbol{]}\}\notag\\
= & N.SNR. \{2+2\mathrm{Re}\boldsymbol{[}e^{j\Delta{\phi}}\mathcal{C}(M_2-M_1)\boldsymbol{]}-|\hat{\boldmath{\gamma_p}}|^2\}. \notag\\
\label{62}
\end{align}

On the other hand, regarding real and imaginary parts of the correlation term in \eqref{59}, the second and third terms inside the brace of \eqref{62} can be written
\begin{equation}
\mathrm{Re}\boldsymbol{[}e^{j\Delta{\phi}}\mathcal{C}(M_2-M_1)\boldsymbol{]} = cos(\Delta{\phi})\cos(L'\alpha)\frac{\sin(NL'\alpha)}{N\sin(L'\alpha)}+sin(\Delta{\phi})\frac{\sin(NL'\alpha)}{N}.
\label{63}
\end{equation}
\begin{align}
|\hat{\boldmath{\gamma_p}}|^2 =&                                                                         2|\mathcal{C}(\frac{M_1-M_2}{2})|^2+2\mathrm{Re}.\boldsymbol{[}e^{-j\Delta{\phi}}\mathcal{C}^2(\frac{M_1-M_2}{2}){]}\notag\\
=&2\frac{\sin^2(\frac{NL'\alpha}{2})}{N^2\sin^2(\frac{L'\alpha}{2})}+2cos(L'\alpha-\Delta{\phi})\frac{\sin^2(\frac{NL'\alpha}{2})}{N^2\sin^2(\frac{L'\alpha}{2})}\notag\\
=&4cos^2(\frac{L'\alpha}{2}-\frac{\Delta{\phi}}{2})\frac{\sin^2(\frac{NL'\alpha}{2})}{N^2\sin^2(\frac{L'\alpha}{2})}
\label{64}
\end{align}
which results in final expression of noncentrality term as
\begin{align}
\lambda_r = & N.SNR. \{2+2cos(\Delta{\phi})\cos(L'\alpha)\frac{\sin(NL'\alpha)}{N\sin(L'\alpha)}+2sin(\Delta{\phi})\frac{\sin(NL'\alpha)}{N} \notag\\
&-4cos^2(\frac{L'\alpha}{2}-\frac{\Delta{\phi}}{2})\frac{\sin^2(\frac{NL'\alpha}{2})}{N^2\sin^2(\frac{L'\alpha}{2})}\}\notag\\
=& N. SNR.\vartheta(\alpha,\Delta{\phi}).
\label{65}
\end{align}
\begin{figure}
 \centering
 \includegraphics[width=0.8\linewidth]{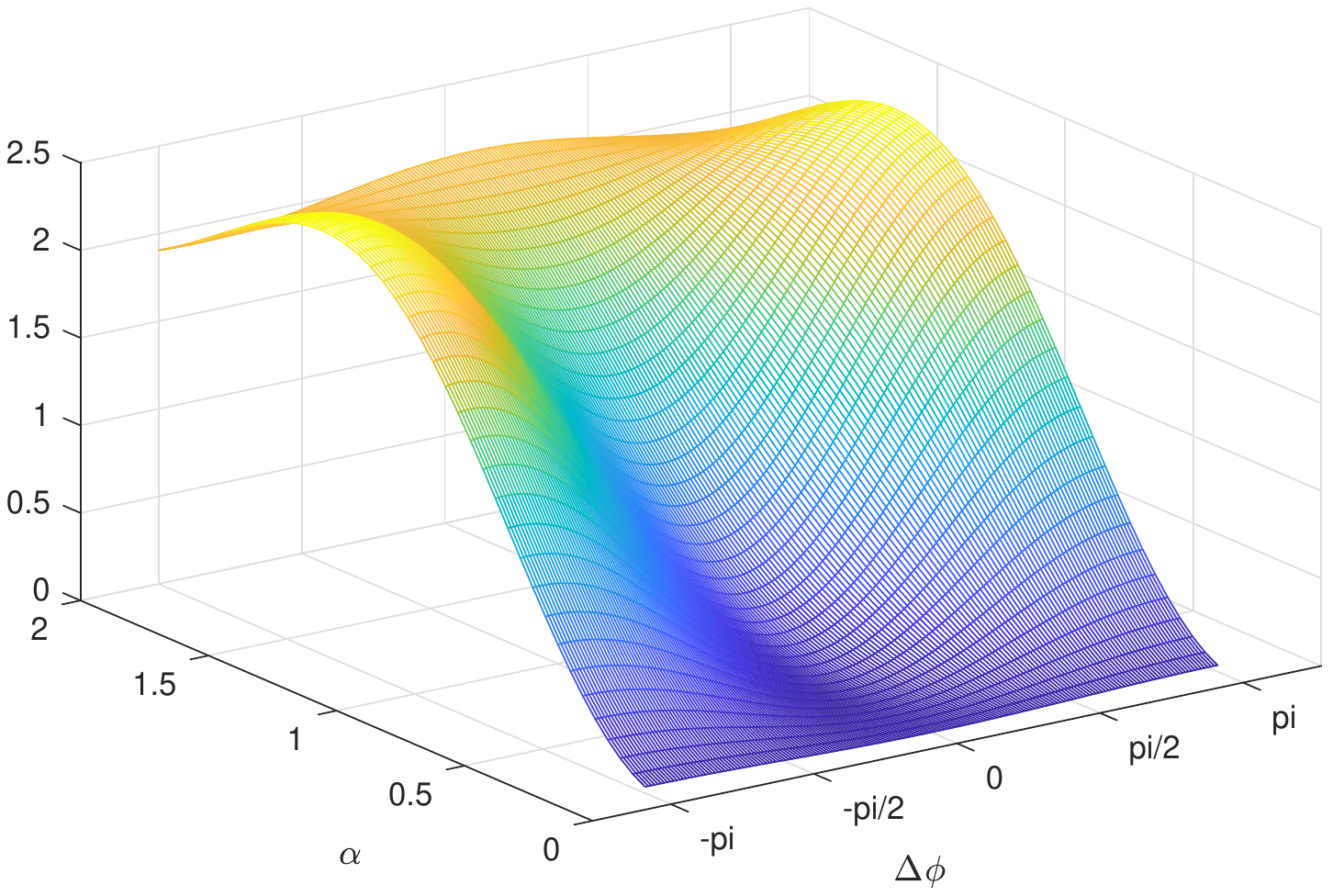}
 \caption{Plot of $\vartheta(\alpha,\Delta{\phi})$.}
 \label{fig:vartheta}
\end{figure}

And the proof is done. The plot of $\vartheta(\alpha,\Delta{\phi})$ is shown in Fig. \ref{fig:vartheta}.

\section{Proof of Proposition \ref{proposition 2}}
\label{App_D}
Based on the assumption, detection of $k$ points is desired in the $k$ detected partial supports. Regarding \eqref{6}, location of targets are estimated as
\begin{equation}
\hat{\Omega}_{k}=\arg\min_{\Omega_{k}}\{\mathbf{g}^H[\mathbf{I}-\mathbf{A}_{\Omega_k}(\mathbf{{A^H}}_{\Omega_k}\mathbf{A}_{\Omega_k})^{-1}\mathbf{{A^H}}_{\Omega_k}]\mathbf{g}\}
\label{66}
\end{equation} 
Let's define $\Omega_k=[m_1,m_2,...,m_k]$ as the $k$-element support whose elements are in a range from 1 to $M$, and  $\mathbf{A}_{\Omega_k}$ is the related $k$-columns steering matrix $\mathbf{A}_{\Omega_k}=[\boldsymbol{a}_{m_1},\boldsymbol{a}_{m_2},... ,\boldsymbol{a}_{m_k}]$.

Regarding $\mathbf{g}^H\mathbf{g}$ which is constant, the minimization problem in \eqref{66} is converted to the following maximization problem
\begin{equation}
\hat{\Omega}_{k}=\arg\max_{\Omega_{k}}\{\mathbf{g}^{H}\mathbf{A}_{\Omega_k}(\mathbf{{A^H}}_{\Omega_k}\mathbf{A}_{\Omega_k})^{-1}\mathbf{{A^H}}_{\Omega_k}\mathbf{g}\}
\label{67}
\end{equation}

Simplifying the middle term of above argument gives
\begin{equation}
\begin{split}
\mathbf{{A^H}}_{\Omega_k}\mathbf{A}_{\Omega_k}
&=
\begin{bmatrix}
\boldsymbol{a}_{m_1}^H\\
\boldsymbol{a}_{m_2}^H\\
\vdots \\
\boldsymbol{a}_{m_k}^H
\end{bmatrix}
\begin{bmatrix}
\boldsymbol{a}_{m_1} & \boldsymbol{a}_{m_2} & \ldots & \boldsymbol{a}_{m_k}
\end{bmatrix}\\
&=\begin{bmatrix}
\boldsymbol{a}_{m_1}^H\boldsymbol{a}_{m_1} & \boldsymbol{a}_{m_1}^H\boldsymbol{a}_{m_2} & \ldots & \boldsymbol{a}_{m_1}^H\boldsymbol{a}_{m_k}\\
\boldsymbol{a}_{m_2}^H\boldsymbol{a}_{m_1} & \boldsymbol{a}_{m_2}^H\boldsymbol{a}_{m_2} & \ldots & \boldsymbol{a}_{m_2}^H\boldsymbol{a}_{m_k}\\
\ldots & \ldots &\ldots &\ldots\\
\boldsymbol{a}_{m_k}^H\boldsymbol{a}_{m_1} & \boldsymbol{a}_{m_k}^H\boldsymbol{a}_{m_2} & \ldots & \boldsymbol{a}_{m_k}^H\boldsymbol{a}_{m_k}
\end{bmatrix}
\end{split}
\label{68}
\end{equation}

To search $k$ points inside $k$ distinct regions, every partial support should include one point unless some of the partial supports do not contain any point which is in contradiction with the definition of detected support. Hence, every partial support includes an element of $\Omega_k$, i.e $m_i \subset supp_i, \forall i \in [1,k]$. On the other hand, regarding the steering matrix which is a partial Fourier matrix, the columns whose spacing is more than $2\rho_s$ are orthogonal. Consequently

\begin{equation}
\boldsymbol{a}_{m_i}^H\boldsymbol{a}_{m_j}=\left\{
\begin{array}{ll}
N, \hspace{.9cm}  i=j\\
\\
0, \hspace{1cm}  i\neq j
\end{array}
\right.
\label{69}
\end{equation}

Thus
\begin{equation}
\begin{split}
\mathbf{{A^H}}_{\Omega_k}\mathbf{A}_{\Omega_k}
=\begin{bmatrix}
N & 0 & \ldots & 0\\
0 & N & \ldots & 0\\
\ldots & \ldots &\ldots &\ldots\\
0 & 0 & \ldots & N
\end{bmatrix}
&=N\mathbf{I}_k
\end{split}
\label{70}
\end{equation}

By substituting \eqref{70} in the maximization problem \eqref{67}, the argument will be
\begin{equation}
\begin{split}
\mathbf{g}^{H}\mathbf{A}_{\Omega_k}(N\mathbf{I}_k)^{-1}\mathbf{{A^H}}_{\Omega_k}\mathbf{g}
&=\frac{1}{N}\mathbf{g}^{H}
\begin{bmatrix}
\boldsymbol{a}_{m_1} & \boldsymbol{a}_{m_2} & \ldots & \boldsymbol{a}_{m_k}
\end{bmatrix}
\begin{bmatrix}
\boldsymbol{a}_{m_1}^H\\
\boldsymbol{a}_{m_2}^H\\
\vdots \\
\boldsymbol{a}_{m_k}^H
\end{bmatrix}
\mathbf{g}\\
&=\frac{1}{N}\mathbf{g}^{H}[(\boldsymbol{a}_{m_1}\boldsymbol{a}_{m_1}^H+\boldsymbol{a}_{m_2}\boldsymbol{a}_{m_2}^H+\ldots+\boldsymbol{a}_{m_k}\boldsymbol{a}_{m_k}^H)]\mathbf{g}\\
&+\frac{1}{N}\mathbf{g}^{H}[\sum\limits_{i=1}^N \sum\limits_{j=1,j\neq{i}}^N \boldsymbol{a}_{m_i}\boldsymbol{a}_{m_j}^H]\mathbf{g}
\end{split}
\label{71}
\end{equation}

Due to orthogonality property of steering vectors mentioned in \eqref{69}, the second term of \eqref{71} becomes zero which yields the following maximization
\begin{equation}
\hat{\Omega}_{k}=\arg\max_{\Omega_{k}}\{\frac{1}{N}\{|\mathbf{g}^H\boldsymbol{a}_{m_1}|^2+|\mathbf{g}^H\boldsymbol{a}_{m_2}|^2+...+|\mathbf{g}^H\boldsymbol{a}_{m_k}|^2\}\}
\label{72}
\end{equation} 

When the partial supports are dis-contiguous, the terms inside the argument of \eqref{72} become independent. Hence, maximization of total expression is equivalent to the maximization of individual terms as 
\begin{equation}
\hat{\Omega}_{k}=\bigcup\limits^{k}_{i=1} {\arg\max_{m_i\in{supp_i}}{|\mathbf{g}^H\boldsymbol{a}_{m_i}|^2}}
\label{73}
\end{equation} 

On the other hand, according to Algorithm \ref{alg.CoarseLocalization}, the detected peaks in the course detection step are obtained by maximization of the canceled terms (residuals). Without loss of generality and for clarity, we consider $k=2$. In this case the two peak points are detected as follows
\begin{equation}
p_1=\arg\max_{i\in{supp_1}}|\mathbf{g}^H\boldsymbol{a}(s_i)|^2
\label{74}
\end{equation} 
\begin{equation}
p_2=\arg\max_{i\in{supp_2}}|\mathbf{g}^H_{\perp}\boldsymbol{a}(s_i)|^2
\label{75}
\end{equation} 

Maximization of the argument of \eqref{74} for $p_1$ is equivalent to maximization of the first term of \eqref{73}. Hence $m_1=p_1$. About $p_2$, the maximization in \eqref{75} can be simplified as
\begin{align}
p_2&=arg\max_{i\in{supp_2}}|\mathbf{g}^H_{\perp}\boldsymbol{a}(s_i)|^2\notag\\
&=arg\max_{i\in{supp_2}}|\boldsymbol{a}^H(s_i)\mathbf{g}_{\perp}|^2\notag\\
&=arg\max_{i\in{supp_2}}|\boldsymbol{a}^H(s_i)(\mathbf{g} - \boldsymbol{a}(s_{p_1})[\boldsymbol{a}(s_{p_1})^H\boldsymbol{a}(s_{p_1})]^{-1}\boldsymbol{a}(s_{p_1})^H\mathbf{g})|^2\notag\\
&=arg\max_{i\in{supp_2}}|\boldsymbol{a}^H(s_i)\mathbf{g} - \boldsymbol{a}^H(s_i)\boldsymbol{a}(s_{p_1})[\boldsymbol{a}(s_{p_1})^H\boldsymbol{a}(s_{p_1})]^{-1}\boldsymbol{a}(s_{p_1})^H\mathbf{g})|^2\notag\\\notag\\
\label{76}
\end{align}

As $s_{p_1}$ and $s_{p_2}$ are far away ($|s_{p_2}-s_{p_1}|>2\rho_s$) due to being in discontinuous partial supports, their  steering vectors are orthogonal i.e $\boldsymbol{a}^H(s_i)\boldsymbol{a}(s_{p_1})=0$. Consequently, the second term inside $|.|$ in \eqref{76} becomes zero and $p_2$ is obtained as 
\begin{equation}
p_2=arg\max_{i\in{supp_2}}|\mathbf{g}^H\boldsymbol{a}(s_i)|^2
\label{77}
\end{equation}
which is equivalent to maximizing the second term of \eqref{73}. Hence, $m_2=p_2$ and the same logic can be applied for the rest of detected peaks when $k>2$, i.e. $m_i=p_i, 1\le{i\leq{k}}$, or equivalently
\begin{equation}
\hat{\Omega}_{k}=\{p_1,p_2,...,p_k\}
\label{78}
\end{equation}


\bibliography{Library}
\bibliographystyle{spiejour}   


\vspace{2ex}\noindent\textbf{Ahmad Naghavi} received his B.Sc. degree in electrical engineering from Amirkabir University of Technology (Tehran Polytechnic), Tehran, Iran in 2001, and the M.Sc. degree from the Ferdowsi University of Mashhad, Iran, in 2004. He is currently a Ph.D. student in the Electrical Engineering department of Isfahan University of Technology (IUT). His research interests include statistical array signal processing, compressed sensing, and TomoSAR processing.

\vspace{2ex}\noindent\textbf{Mohammad Sadegh Fazel} received his B.Sc. degree in electrical engineering from Isfahan University of Technology
(IUT), Isfahan, Iran in 1996, and the M.Sc. degree from the University of Tehran, Iran, in 1998, and the
Ph.D. degree from Institute for Communication Systems, University of Surrey, UK in 2010. Since 2011, he is with the Electrical Engineering department of Isfahan University of Technology (IUT). His current research interests include 6G wireless communications, massive MIMO, cooperative relaying, and spectral efficiency.

\vspace{2ex}\noindent\textbf{Mojtaba Beheshti} received the B.Sc. degree from the Isfahan University of Technology (IUT), Iran, the M.Sc. degree from the University of Tehran, and the Ph.D. degree from IUT in 1996, 1999, and 2011 respectively, all in electrical engineering. Currently, he is an assistant professor with the Information and Communication Technology Institute of IUT. His research interests include signal processing for synthetic aperture radar, multicarrier and Massive MIMO systems.

\vspace{2ex}\noindent\textbf{Ehsan Yazdian} received his B.Sc. degree in electrical engineering from Isfahan University of Technology
(IUT), Isfahan, Iran in 2004, and the M.Sc. and Ph.D.
degree in electrical engineering from the Sharif University of Technology, Tehran, Iran, in 2006 and 2012,
respectively. Since 2013, he is with the Electrical Engineering department of Isfahan University of Technology (IUT). His research interests are in the areas of statistical array signal processing, wireless communications, digital communication systems, and software-defined radio.

\end{spacing}
\end{document}